\documentclass[sn-mathphys-ay]{sn-jnl}

\usepackage{graphicx}%
\usepackage{multirow}%
\usepackage{amsmath,amssymb,amsfonts}%
\usepackage{amsthm}%
\usepackage{mathrsfs}%
\usepackage{bm}
\usepackage{float}
\usepackage{subfigure}
\usepackage[title]{appendix}%
\usepackage{xcolor}%
\usepackage{textcomp}%
\usepackage{tabularx}
\usepackage{manyfoot}%
\usepackage{booktabs}%
\usepackage{algorithm}%
\usepackage{algorithmicx}%
\usepackage{algpseudocode}%
\usepackage{listings}%
\usepackage{arydshln}

\theoremstyle{thmstyleone}%
\theoremstyle{thmstyletwo}%
\newtheorem{thm}{\bf Theorem}      
\newtheorem{cor}{\bf Corollary}[section]     
\newtheorem{lem}{\bf Lemma}
\newtheorem{prop}{\bf Proposition}

\newtheorem{rem}{\bf Remark}

\theoremstyle{thmstylethree}%

\begin{document}

\title[Feature splitting parallel algorithm for Dantzig selectors]{Feature splitting parallel algorithm for Dantzig selectors}


\author[1]{\fnm{Xiaofei} \sur{Wu}}\email{xfwu1016@163.com}
\equalcont{These authors contributed equally to this work.}
\author[2]{\fnm{Yue} \sur{Chao}}\email{cyasy@xmu.edu.cn}
\equalcont{These authors contributed equally to this work.}
\author[3]{\fnm{Rongmei} \sur{Liang}}\email{liang\_r\_m@163.com}
\author[3]{\fnm{Shi} \sur{Tang}}\email{loveqq200305@163.com}
\author*[1]{\fnm{Zhiming} \sur{Zhang}}\email{zmzhang@cqu.edu.cn}


\affil[1]{\orgdiv{College of Mathematics and Statistics}, \orgname{Chongqing University}, \orgaddress{\street{Huxi}, \city{Chongqing}, \postcode{401331}, \state{Chongqing}, \country{China}}}

\affil[2]{\orgdiv{Department of Statistics and Data Science}, \orgname{Xiamen University}, \orgaddress{\street{Siming}, \city{Xiamen}, \postcode{361005}, \state{Fujian}, \country{China}}}

\affil[3]{\orgdiv{Big Data and Intelligence Engineering School}, \orgname{Chongqing College of International Business and Economics}, \orgaddress{\street{Xuefu}, \city{Chongqing}, \postcode{401520}, \state{Chongqing}, \country{China}}}



\abstract{The Dantzig selector is a widely used and effective method for variable selection in ultra-high-dimensional data. Feature splitting is an efficient processing technique that involves dividing these ultra-high-dimensional variable datasets into manageable subsets that can be stored and processed more easily on a single machine. This paper proposes a variable splitting parallel algorithm for solving both convex and nonconvex Dantzig selectors based on the proximal point algorithm. The primary advantage of our parallel algorithm, compared to existing parallel approaches, is the significantly reduced number of iteration variables, which greatly enhances computational efficiency and accelerates the convergence speed of the algorithm. Furthermore, we show that our solution remains unchanged regardless of how the data is partitioned, a property referred to as partition-insensitive.  In theory, we use a concise proof framework to demonstrate  that the algorithm exhibits linear convergence. Numerical experiments indicate that our algorithm performs competitively in both parallel and nonparallel environments. The R package  for implementing the proposed algorithm can be obtained at \url{https://github.com/xfwu1016/PPADS}. }


\maketitle

\section{Introduction}\label{sec1}

As scientific and technological advancements progress, data collection has become increasingly easier. One manifestation of this is the growing number of features (variables) in datasets, even though the number of truly useful features remains limited. Consequently, many variable selection methods have been proposed, such as subset selection, stepwise regression, and regularization techniques.
Among these, the Dantzig selector (abbreviated as DS), introduced by \cite{Candes2007The}, is one of the most well-known methods for ultra-high-dimensional data (\textcolor{red}{the number of features far exceeds the number of samples}), defined as follows,
\begin{align}\label{ods}
\min\limits_{\bm \beta \in \mathbb{R}^p} \quad \lVert \bm \beta \rVert_1 \quad \text{s.t.} \quad  \lVert  \bm X^\top(\bm X \bm \beta - \bm y)/n \rVert_{\infty} \leq \lambda,
\end{align}
where 
\[
\bm{X} = \begin{pmatrix} \bm{X}_1 & \ldots & \bm{X}_p \end{pmatrix} \  \text{and} \
\bm{y} = \begin{pmatrix} y_1 & \ldots & y_n \end{pmatrix}^\top
\]
denote the \(n \times p\) design matrix and the response vector, respectively; $\lambda$ denotes a non-negative tuning parameter. Here, we write $\|\mathbf{x}\|_q$ for the $\ell_q$ norm of $\mathbf{x} \in \mathbb{R}^p$, where $1 \leq q \leq \infty$. 
\textcolor{red}{
We operate under the assumption of a  linear regression model \(\bm y = \bm X \bm \beta+ \bm \epsilon\), where \( \bm \epsilon\) represents the error vector. The true regression coefficient vector \(\beta\) remains generally unknown.
The parameter \(\lambda\) is of paramount importance. A smaller \(\lambda\) tightens the constraint \(\lVert \bm X^\top(\bm y - \bm X \bm \beta)\rVert_{\infty}\leq\lambda\), resulting in a sparser estimator. In contrast, a larger \(\lambda\) relaxes the constraint, potentially leading to the DS  estimator having more non-zero elements.
DS was devised to tackle the variable selection and estimation challenges in  ultra-high-dimensional linear regression. When the number of features \(p\) significantly surpasses the sample size \(n\), traditional linear regression methods become ineffective. DS overcomes this limitation by integrating an \(\ell_1\) norm constraint and an inequality constraint, enabling efficient variable selection and estimation. Moreover, some nonconvex DS models in \cite{Wen2024Nonconvex}  possesses the Oracle property, signifying that under specific conditions, it can identify the true variable set with a probability approaching 1, and the estimated coefficients exhibit asymptotic normality. Next, we will review the theoretical and algorithmic work on DS and some of its variants. }

The optimal $\ell_2$ rate properties for estimator of DS were established under a sparsity scenario, and impressive empirical performance on real-world problems involving large values of $p$ was demonstrated in \cite{Candes2007The}. Due to its outstanding performance in managing ultra-high dimensional data and its tight error bounds, the Dantzig selector (DS) has attracted considerable attention. Discussions on the Dantzig selector can be found in the works of \cite{Bickel2007Discussion}, \cite{Cai2007Discussion}, \cite{Cands2007Rejoinder}, \cite{Efron2007Discussion}, \cite{Friedlander2007Discussion}, \cite{Meinshausen2007Discussion}, and \cite{Ritov2007Discussion}. In addition, the approximate equivalence relationship between DS and another popular variable selection method Lasso (\cite{Tibshirani1996Regression}) was established by \cite{Bickel2008Simultaneous} and \cite{James2008DASSO}. 

Many studies on DS variants have emerged in recent years. \cite{Dicker2013Parallelism} proposed the adaptive DS to enhance the selection performance of the original DS by incorporating the adaptive weighting concept introduced by \cite{Zou2006The}. \cite{Liu2010The} suggested using group regularization terms instead of $\ell_1$ regularization to develop a group DS, allowing it to better adapt to data containing grouping prior information. \cite{Chatterjee2014Generalized} introduced a generalized DS for linear models, where any norm that captures the parameter structure can be utilized for estimation. Recently, \cite{Ge2022The} and \cite{Wen2024Nonconvex} proposed the nonconvex DSs based on \(\ell_1 - \alpha \ell_2\) regularization, SCAD and MCP. DS has also been extended to other model settings, including generalized linear models in \cite{James2008DASSO}, censored linear regression models in \cite{Li2014The}, and multiple quantile regression models in \cite{Park2017Dantzigtype}.

As discussed by \cite{Candes2007The}, a natural approach to solving (\ref{ods}) is to reformulate it as a linear programming (LP) problem using a primal–dual interior point algorithm (\cite{Boyd2004Convex}). However, interior point methods are typically inefficient for high-dimensional problems ($p>n$), as they require solving dense Newton systems at each iteration. An alternative approach for solving (\ref{ods}) involves using homotopy methods to compute the entire solution path of the Dantzig selector, see \cite{James2008DASSO}. Nonetheless, as noted in \cite{Becker2010Templates} (Section 1.2), these methods also struggle with  high-dimensional problems. To enable the DS to be applicable to high-dimensional datasets, \cite{Becker2010Templates} reformulated (\ref{ods}) as a linear cone programming problem, for which a smooth approximation of its dual problem is solved using a projected gradient algorithm (\cite{Beck2009A}).

Observing (\ref{ods}), it contains \(\ell_1\) and \(\ell_\infty\) terms that are not smooth. A widely applicable algorithm for addressing this is the Alternating Direction Method of Multipliers (ADMM, \cite{Boyd2010Distributed}). \cite{Lu2012An} was the first to  rewrite (\ref{ods}) into an optimization form with separable equality constraints, and utilize ADMM to find the Dantzig selector. Numerical experiments in  \cite{Lu2012An} demonstrated that this method generally outperforms the approach presented by \cite{Becker2010Templates} in terms of CPU time while producing solutions of comparable quality. It is worth noting that ADMM is an iterative algorithm, and the efficiency of the algorithm depends on whether each subproblem can be efficiently solved.  In each iteration of ADMM described by \cite{Lu2012An}, two subproblems needed to be solved successively. One of the subproblems had a closed-form solution, while the other did not and was approximated using a nonmonotone gradient method (NGM). However, NGM is also an iterative algorithm, and embedding it within ADMM can lead to double loops, resulting in excessive computational burden. To address the challenges arising from the subproblem with double loops, \cite{Wang2012The} proposed a linearized version of ADMM (LADMM) specifically tailored for the Dantzig selector. This modified approach demonstrated superior efficiency compared to ADMM in \cite{Lu2012An} when solving both synthetic and real-world datasets. \cite{Li2015The} encapsulated the LADMM algorithm in an efficient and user-friendly R package called ``flare". In addition, the ADMM algorithm and LADMM have been employed by \cite{Ge2022The} and \cite{Chatterjee2014Generalized} to solve the \(\ell_1 - \alpha \ell_2\) DS and the generalized DS, respectively.

In addition to the above two popular ADMM algorithms, several other algorithms have been proposed to solve the DS. With the aid of proximal operators, \cite{Prater2015Finding} skillfully introduced two simple iterative methods through a fixed-point reformulation of the solutions to (\ref{ods}). They noted that their first iterative algorithm followed the same steps as the LADMM iteration presented by \cite{Wang2012The}, while the second iterative algorithm aligned with the steps of the primal-dual iteration algorithm introduced by \cite{Chambolle2011A} for DS. 
A fast splitting method introduced in \cite{He2015A} shows competitive performance when compared to both ADMM and LADMM. Subsequently, a series of other splitting algorithms, including partially linearized ADMM and customized proximal point algorithms (CPPA, proposed by \cite{Gu2014Customized}), has been proposed in \cite{He2017Splitting} to tackle the DS.  Recently, \cite{Fang2021Efficient} reformulated DS problem as an equivalent convex composite optimization problem and proposes a semismooth Newton augmented Lagrangian  algorithm to solve it, while also applying a proximal point dual semismooth Newton algorithm for another equivalent form of the problem. However, the iterative algorithms mentioned above were designed by formulating the optimization of DS as a constrained optimization problem. During the solution process of each subproblem, it was necessary to introduce dual variables to eliminate the constraints. \cite{Mao2021A} was the first to reformulate DS as an unconstrained optimization problem and proposed a partially proximal linearized alternating minimization method (P-PLAM). Notably, its iterative process does not require updates of dual variables, and two subproblems have straightforward closed-form solutions.

It is worth noting that the aforementioned algorithms do not consider the use of variable splitting and parallel computation frameworks to enhance the efficiency of DS in handling high-dimensional data. More recently, inspired by the works of  \cite{Sun2015A},  \cite{Li2023Regularized},  \cite{Wen2023Feature} and \cite{Cai2024Asset},  \cite{Wen2024Nonconvex} proposed an efficient computational algorithm for solving DS and nonconvex DSs, utilizing a feature splitting approach. This approach not only reduces memory usage but also facilitates the parallel implementation of the computational framework. This parallel algorithm is based on the three-block ADMM algorithm proposed by \cite{Sun2015A}. However, as noted by \cite{Chen2016The}, the directly extended three-block ADMM algorithm cannot theoretically guarantee convergence. To ensure convergence, \cite{Sun2015A} introduced updates for certain intermediate variables at each iteration. Although this algorithm has theoretical convergence properties, its application within a parallel computing structure results in excessive redundant intermediate variables. As pointed out by \cite{Lin2022Alternating}, these excessive redundant intermediate variables in ADMM algorithm can degrade solution precision and slow down algorithm convergence. Specifically, the parallel algorithm proposed by \cite{Wen2024Nonconvex} necessitates updating $3Kp$ variables per iteration, where $K$ represents the number of variable partitions. Clearly, in the context of ultra-high dimensions, if $K$ is notably large, the number of variables that the algorithm needs to update becomes excessive.

This paper aims to propose a unified optimization algorithm for solving both DS and nonconvex DS models, effectively mitigating the issue of excessive intermediate variables associated with the three-block ADMM in \cite{Wen2024Nonconvex}.  This new parallel algorithm is framed within the PPA (proximal point algorithms in \cite{Parikh2013Proximal})  structure.  Compared to the three-block ADMM algorithm in \cite{Wen2024Nonconvex}, our algorithm has the following three significant advantages.
\begin{enumerate}
\item The proposed algorithm requires updating only $3p$ variables in each iteration, which is independent of the number of variable partitions $K$. Fewer updated variables not only enhance the computational speed of each iteration, accelerating algorithm convergence, but also result in solutions with higher precision.

\item Our parallel algorithm exhibits partitioning insensitivity through specific structural design. Specifically, in parallel environment, the  solution of algorithm remains invariant regardless of how the matrix is partitioned.  The partition insensitivity ensures the reliability and accuracy of solutions in a parallel computing environment, independent of the chosen parallelization strategy. This facilitates a more flexible adaptation of computing resources to accommodate diverse computational requirements and hardware setups.

\item We provide a concise proof demonstrating that the proposed algorithm achieves a linear convergence rate. The linear convergence speed of the algorithm ensures stable progress and efficient computational performance, enabling reliable approximation of the exact solution in a shorter amount of time.
\end{enumerate}

The structure of this paper is organized as follows. In Section \ref{sec2}, we review relevant algorithms for solving DS. Section \ref{sec3} discusses two parallel PPA algorithms for solving DS and its nonconvex variants, including their insensitivity to algorithm partitioning and convergence properties. Synthetic numerical simulations are presented in Section \ref{sec4}, while real data is discussed in Section \ref{sec5}.  Section \ref{sec6} summarizes the findings and concludes the paper, with a discussion on future research directions. Technical proofs and supplementary experiments are included in the Appendix. The R package  for implementing the proposed algorithm can be obtained at \url{https://github.com/xfwu1016/PPADS}.

\textbf{Notations}:
$\bm 0_n$ and $\bm 1_n$ represent $n$-dimensional vectors with all elements being 0 and 1, respectively.
$\bm I_n$ represents the $n$-dimensional identity matrix.
The Hadamard product is denoted by $\odot$.
The sign$(\cdot)$ function is defined component-wise such that sign$(t) = 1$ if $t > 0$, sign$(t) = 0$ if $t = 0$, and sign$(t) = -1$ if $t < 0$.
$(\cdot)_+$ signifies the element-wise operation of extracting the positive component, while $|\cdot|$ denotes the element-wise absolute value function.
For any vector $\bm u$, $\|\bm u\|_1$, $\|\bm u\|_2$, and $\|\bm u\|_{\infty}$ denote the $\ell_1$ norm, the $\ell_2$ norm, and the $\ell_{\infty}$ norm of $\bm u$, respectively.
$\|\bm u\|_{\bm H} := \sqrt{\bm u^\top \bm H \bm u}$ is used to denote the norm of $\bm u$ under the matrix $\bm H$, where $\bm H$ is a matrix. If \( \bm H \) is positive semi-definite, we denote it as \( \bm H  \succeq \bm 0 \); if the matrix is positive definite, we denote it as \( \bm H  \succ \bm 0 \).

\section{Review of Related Algorithms}\label{sec2}
As the focus of this paper is on parallel computation, currently, the only algorithm capable of implementing parallel computation for the DS optimization objective function is ADMM in \cite{Wen2024Nonconvex}. Therefore, this paper solely reviews some ADMM algorithms. For detailed introductions to other first-order and second-order algorithms for computing DS, readers can refer to the literature on algorithms mentioned in the introduction of this paper. ADMM algorithm is an iterative optimization method designed to tackle complex convex minimization problems with linear constraints. It operates by decomposing the original problem into smaller, more manageable subproblems that are easier to solve. ADMM alternates between optimizing these subproblems and updating dual variables to enforce the constraints, making it particularly effective for large-scale problems and various statistical learning applications. For a comprehensive review, readers can refer to \cite{Boyd2010Distributed}.

In order to make (\ref{ods}) satisfy the condition that ADMM can handle problems, we need to introduce the auxiliary variable $\bm X^\top(\bm X \bm \beta - \bm y)  =  \bm z$.  To simplify the symbol, we make $\bm A = \bm X^\top \bm X$.  Then, (\ref{ods}) can be rewritten into the following equation constrained optimization form, 
\begin{equation}\label{ds}
\begin{aligned}
& \min_{\bm \beta, \bm z} \quad   \|\bm \beta\|_{1} +  \delta_{\mathcal{Z}_0(\bm z)}  \\
& \text{s.t.} \quad \bm A \bm \beta -  \bm z = \bm X^\top \bm y.
\end{aligned}   
\end{equation}
where  \( \mathcal{Z}_0(\bm z) = \left\{ \bm z : \| \bm z \|_{\infty} \leq n\lambda \right\}\) and \(
\delta_{\mathcal{Z}_0(\bm z)}  =
\begin{cases}
0, & \text{if } \bm z \in \mathcal{Z}_0, \\
+\infty, & \text{otherwise.}
\end{cases}
\)

The augmented Lagrangian form of (\ref{ds}) is
\begin{align}\label{intr1}
L_\mu(\bm \beta, \bm z; \bm u) = 
\| \bm \beta \|_1  + \delta_{\mathcal{Z}_0(\bm z)}  + \bm u^\top (  \bm A \bm \beta  -  \bm z   -  \bm X^\top \bm y ) +  \frac{\mu}{2}  \|    \bm A \bm \beta  -  \bm z   -  \bm X^\top \bm y \|_2^2,
\end{align}
where $\bm u$ is the dual variable corresponding to the linear constraint, and $\mu>0$ is a given  augmented parameter.

\subsection{ADMM and LADMM for DS}
Given \(( \bm{z}^0, \bm{u}^0)\), the iterative scheme of ADMM for problem (\ref{intr1}) is as follows,
\begin{equation}\label{twoupadmm}
\begin{aligned}
&\bm \beta^{t+1} \ \leftarrow  \mathop {\arg \min }\limits_{\bm \beta} \left\{ L_\mu(\bm \beta,  \bm z^t; \bm u^t)  \right \};\\
&\bm z^{t+1} \ \leftarrow  \mathop {\arg \min }\limits_{\bm z} \left\{L_\mu(\bm \beta^{t+1},  \bm z;  \bm u^t)  \right \}; \\
&\bm u^{t+1} \ \leftarrow  \bm{u}^{t}  + \mu( \bm A  \bm \beta^{t+1}  - \bm z^{t+1} - \bm X^\top \bm y).
\end{aligned}     
\end{equation}
Note that the ADMM iteration order mentioned in \cite{Lu2012An} is \(\bm{z} \to \bm{\beta} \to \bm{u}\). However, the order of the iterations does not affect the convergence or efficiency of the ADMM algorithm. The choice of the iteration order in (\ref{twoupadmm}) is made to maintain consistency with the LADMM presented in \cite{Wang2012The} and the parallel ADMM described in \cite{Wen2024Nonconvex}.

For the three subproblems of the ADMM iteration, the \(\bm \beta\) subproblem
\begin{align}\label{admmbeta}
\bm \beta^{t+1} \ \leftarrow  \mathop {\arg \min }\limits_{\bm \beta} \left\{ \| \bm \beta \|_1   +  \frac{\mu}{2}  \|    \bm A \bm \beta  -  \bm z^t   -  \bm X^\top \bm y  + \frac{\bm u^t}{\mu}\|_2^2 \right \},
\end{align}
which is similar to a Lasso problem and does not have a closed-form solution. \cite{Lu2012An} suggests using a nonmonotone gradient method to iteratively solve it. The \(\bm z\) subproblem has a closed-form solution (see Proposition 2.1 in \cite{Lu2012An}), while the \(\bm u\) subproblem involves a linear operation that is easy to implement. \cite{Wang2012The} pointed out that the lack of a closed-form solution in the \(\bm \beta\) step would impact the overall efficiency of the ADMM algorithm iteration. Therefore, it is recommended to use linearization methods to approximate this iteration step, ensuring that this subproblem yields a closed-form solution. Specifically, this involves adding a quadratic term to the objective function of the optimization problem for the \(\bm \beta\) subproblem, as follows, 
\begin{align}\label{ladmmbeta}
\bm \beta^{t+1} \ \leftarrow  \mathop {\arg \min }\limits_{\bm \beta} \left\{ \| \bm \beta \|_1   +  \frac{\mu}{2}  \|    \bm A \bm \beta  -  \bm z^t   -  \bm X^\top \bm y  + \frac{\bm u^t}{\mu}\|_2^2 + \frac{1}{2} \|\bm \beta - \bm \beta^t \|_{\bm S}^2 \right \},
\end{align}
where  $\bm S = \eta \bm I_p - \mu \bm A^\top \bm A$ and  $\eta$ is larger than the maximum eigenvalue of $\mu \bm A^\top \bm A$.  Rewrite the optimization formula above to eliminate some constant terms unrelated to the optimization variable  \(\bm \beta\), and the subproblem of  \(\bm \beta\) becomes
\begin{align}\label{2.3}
\bm \beta^{t+1} \ \leftarrow  \mathop {\arg \min }\limits_{\bm \beta} \left\{\| \bm \beta \|_1 + \frac{\eta}{2} \|\bm \beta - \bm \beta^t + \frac{\mu}{\eta} \bm A^\top (\bm A \bm \beta^t -  \bm z^t   -  \bm X^\top \bm y  + \frac{\bm u^t}{\mu})  \|_{2}^2 \right \}.
\end{align}
Clearly, the above objective function for optimization is a soft-thresholding operator and has a closed-form solution.

In recent years, the linearization technique of the ADMM algorithm has garnered significant attention and been applied in various fields, such as \cite{Li2014Linearized}, \cite{Gu2018ADMM}, \cite{Liang2024Linearized}, and \cite{Wu2025Multi}. It is worth mentioning that \cite{Liang2024Linearized} provides an efficient power method for computing the maximum eigenvalue of a positive semidefinite matrix in the design of LADMM. The power method itself is an iterative algorithm, but each iteration only involves matrix-vector multiplications, without matrix-matrix multiplications. In terms of numerical performance, it typically converges in a few dozen steps at most. In LADMM, note that \(\eta\) does not require a very accurate solution; it only needs to be greater than the maximum eigenvalue of the positive semidefinite matrix. Therefore, under relaxed convergence conditions, the power method can converge more quickly.  In this paper, we recommend using this efficient power method  to compute the maximum eigenvalue of all positive semidefinite matrices. 

\subsection{Parallel ADMM for DS}\label{sec22}
In a parallel computing environment, the estimation coefficients of DS can be divided into  \(K\) blocks, represented as 
$\bm \beta = (\bm \beta_{1\cdot}^\top, \bm \beta_{2\cdot}^\top, \cdots, \bm \beta_{K\cdot}^\top )^\top$,
where each $\bm \beta_{i\cdot} \in \mathbb{R}^{p_i}$ and  $\sum_{i=1}^K p_i = p$.
Correspondingly, we can partition the gram matrix \(\bm A = \bm X^\top \bm X\) into  $\bm A = (\bm A_1, \bm A_2,\cdots, \bm A_K)$, where \(\bm A_i \in \mathbb{R}^{p \times p_i}\). \cite{Wen2024Nonconvex} introduced slack variables \(\boldsymbol{\omega} = (\bm \omega_2, \cdots, \bm \omega_K)\), where \(\bm \omega_i \in \mathbb{R}^p\) for \(i \in \{2, \cdots, K\}\), to reformulate the optimization problem as follows:
\begin{align}\label{pds}
\begin{aligned}
\min_{\bm \beta, \bm z, \boldsymbol{\omega}} &\quad \sum_{i=1}^{K}  \|\bm \beta_i\|_{1} +  \delta_{\mathcal{Z}_0(\bm z)}  \\
 \text{s.t.} &\quad \bm A_1 \bm \beta_{1\cdot} + \bm \omega_2 + \cdots + \bm \omega_K -  \bm z = \bm X^\top \bm y, \\
&\quad \omega_i = \bm A_i \bm \beta_i, \quad i \in \{2, \ldots, K\}.
\end{aligned}
\end{align}
 The augmented Lagrangian form of (\ref{pds}) is
\begin{align}
    L_\mu(\bm \beta, \bm z, \bm \omega;  \bm u) & = 
\sum_{i=1}^{K}  \|\bm \beta_{i\cdot}\|_1   + \delta_{\mathcal{Z}_0(\bm z)}  + \bm u_1^\top (  \bm A_1 \bm \beta_{1\cdot}  + \sum_{i=2}^ {K}\bm \omega_i -  \bm z   -  \bm X^\top \bm y ) \\\nonumber
& +  \frac{\mu}{2}  \|   \bm A_1 \bm \beta_{1\cdot}  + \sum_{i=2} ^{K}\bm \omega_i -  \bm z   -  \bm X^\top \bm y \|_2^2\\\nonumber
& + \sum_{i=2}^{K} \left[ \bm u_i^\top  \left( \bm A_i  \bm \beta_{i\cdot} - \bm \omega_i \right) +  \frac{\mu}{2}  \|    \bm A_i  \bm \beta_{i\cdot} - \bm \omega_i  \|_2^2 \right].
\nonumber
\end{align}

It is easy to see that the separability of the \( K \) blocks of \( \bm \beta \) in the augmented Lagrangian function enables a natural parallelization for updating \(\bm \beta \). Similar to the iterative steps of \cite{Sun2015A} and \cite{Wen2023Feature}, \cite{Wen2024Nonconvex} uses the following iterative procedure,
\begin{equation}
\left\{\begin{aligned}\notag
&\bm \beta_{1\cdot}^{t+1}  \leftarrow  \mathop {\arg \min }\limits_{\bm \beta_{1\cdot}} \left\{ \|\bm \beta_{1\cdot}\|_1 + \frac{\mu}{2}  \|   \bm A_1 \bm \beta_{1\cdot}  + \sum_{i=2} ^{K}\bm \omega_i^t -  \bm z^t   -  \bm X^\top \bm y  + \frac{\bm u_1^t}{\mu} \|_2^2 +  \frac{1}{2} \|\bm \beta_{1\cdot} - \bm \beta_{1\cdot}^t \|_{\bm S_1}^2  \right \};\\ %
&\bm \beta_{i\cdot}^{t+1}  \leftarrow  \mathop {\arg \min }\limits_{\bm \beta_{i\cdot}} \left\{ \|\bm \beta_{i\cdot}\|_1 + \frac{\mu}{2}  \|    \bm A_i  \bm \beta_{i\cdot} - \bm \omega_i^t + \frac{\bm u_i^t}{\mu} \|_2^2 +  \frac{1}{2} \|\bm \beta_{i\cdot} - \bm \beta_{i\cdot}^t \|_{\bm S_i}^2  \right \} , i \in \{2, \cdots, K\};\\ %
&\bm \omega_i^{t+\frac{1}{2}}  \leftarrow  (\bm X^\top \bm y + \bm z^t  + K  \bm A_i \bm \beta_{i\cdot}^{t+1} - \bm A \bm \beta^{t+1} ) /K,  i \in \{2, \cdots, K\}; \\ %
&\bm z^{t+1}  \leftarrow  \mathop {\arg \min }\limits_{\bm z} \left\{ \delta_{\mathcal{Z}_0(\bm z)}  +   \frac{\mu}{2}  \|   \bm A_1 \bm \beta_{1\cdot}^{t+1}  + \sum_{i=2} ^{K}\bm \omega_i^{t+\frac{1}{2}} -  \bm z   -  \bm X^\top \bm y  + \frac{\bm u_1^t}{\mu} \|_2^2  \right \}; \\%
&\bm \omega_i^{t+1}  \leftarrow  (\bm X^\top \bm y + \bm z^{t+1}  + K  \bm A_i \bm \beta_{i\cdot}^{t+1} - \bm A \bm \beta^{t+1} ) /K,  i \in \{2, \cdots, K\}; \\ %
&\bm u_1^{k+1} \leftarrow  \bm{u}_1^{k}  + \mu(\bm A_1 \bm \beta_{1\cdot}^{t+1}  + \sum_{i=2}^ {K}\bm \omega_i^{t+1} -  \bm z^{t+1}   -  \bm X^\top \bm y );\\%
&\bm u_i^{k+1} \leftarrow   \bm{u}_i^{k} + \mu \left( \bm A_i  \bm \beta_{i\cdot}^{t+1} - \bm \omega_i^{t+1} \right),  i \in \{2, \cdots, K\}. 
\end{aligned} \right.
\end{equation}

The updates for \(\bm \beta_{1\cdot}\) and \(\bm \beta_{i\cdot}\) are carried out after linearization, where \(\bm S_i = \eta_i \bm I_p - \mu \bm A_i^\top \bm A_i\), with \(\eta_i\) required to exceed the maximum eigenvalue of \(\mu \bm A_i^\top \bm A_i\). It is not difficult to observe that, similar to (\ref{2.3}), the updates for  \(\bm \beta_{1\cdot}\) and \(\bm \beta_{i\cdot}\) have closed-form solutions. The update for \(\bm z\) is analogous to (\ref{twoupadmm}), and a closed-form solution can also be derived based on Proposition 2.1 from \cite{Lu2012An}. The updates of $\bm \omega_i$ and $\bm u_i$ are simple and easy to implement algebraic operations. Clearly, $\bm \omega_i$ and $\bm u_{i}$ are ultra-high dimensional  vector constructed to achieve parallel computing effects. Thus, \( \bm{\omega} \) (including all \( \bm{\omega}_i \)) is a \((K-1)p\)-dimensional vector, while \( \bm{u} \) (including all \( \bm{u}_i \)) is a \( Kp \)-dimensional vector. Due to its excessively large dimensionality, it reduces the convergence speed of the algorithm during implementation and lowers the precision of the solution.

\section{Proximal Point Algorithms for DS}\label{sec3}
Proximal Point Algorithms (PPA) are iterative optimization methods used to solve convex optimization problems, particularly those involving non-smooth functions. The core idea of these algorithms is to replace the original problem with a series of simpler subproblems that can be more easily solved using the proximal operator. The ability of PPA to handle non-smooth losses and the iterative form of variable splitting are similar to those of the ADMM algorithm. \cite{Parikh2013Proximal} offers a comprehensive overview and introduction to PPA and proximal operators, while \cite{Cai2013A} delves into a detailed discussion on the relationship between PPA and ADMM algorithms.

However, in the fields of statistics and machine learning, PPA is far less well-known than ADMM. This difference may be attributed to three main reasons: first, comprehensive review articles on ADMM (\cite{Boyd2010Distributed}) were published earlier; second, PPA and ADMM have a high degree of overlap in their applicable scenarios; and finally, in practice, PPA iterative forms tend to be more flexible and varied, whereas the alternating iteration format of the ADMM algorithm is more widely recognized. Using the two PPAs for solving DS proposed by \cite{He2017Splitting} (Algorithms 4 and 5 in their paper) as examples, these PPAs consist of two iterative steps. The first step, akin to the ADMM iteration, is called the prediction step, while the second step is referred to as the correction step, which essentially serves as a relaxation step. However, the prediction phase does not have a fixed iteration sequence like the ADMM algorithm (where one variable is held constant while the other is updated); instead, it allows for a variable iteration order. In their Algorithm 4, the dual variables are updated first, followed by the primal variables, whereas in  their Algorithm 5, the primal variables are updated first, followed by the dual variable.
This variability, while making the PPA approach more flexible for solving problems, also hinders its application in other disciplines due to its lack of memorability.

In this paper, in order to make PPA easier to remember and use, we adopt an iterative approach similar to the ADMM algorithm, which is a primal-dual alternating iteration. Next,  we will demonstrate that the PPA is more suitable than ADMM for solving DS in the context of parallel computation. Before introducing the PPA algorithm proposed in this paper, we  first review the constrained optimization form of DS,
\begin{align}\label{ds2}
\begin{aligned}
& \min_{\bm \beta, \bm z} \quad   \|\bm \beta\|_{1} +  \delta_{\mathcal{Z}_0(\bm z)}  \\
& \text{s.t.} \quad \bm A \bm \beta -  \bm z = \bm X^\top \bm y.
\end{aligned}
\end{align}
The Lagrange multiplier form of (\ref{ds2}) is   
\[
L(\bm \beta, \bm z;  \bm u) = 
\|\bm \beta\|_1   + \delta_{\mathcal{Z}_0(\bm z)}  - \bm u^\top (  \bm A \bm \beta  -  \bm z   -  \bm X^\top \bm y ),\]
where $\bm u$  is  a Lagrange multiplier (dual variable).

\subsection{Nonparallel version}
With the Lagrangian form of DS, we can introduce our PPA method in the case where the features are not split.
With given \(( \bm \beta^0, \bm{z}^0, \bm{u}^0)\), the iterative scheme of the PPA for problem (\ref{ds2}) is as follows:
\begin{equation}\label{nppa}
\begin{aligned}
&\bm \beta^{t+1} \leftarrow \mathop {\arg \min }\limits_{\bm \beta} \left\{ L(\bm \beta,  \bm z^t; \bm u^t)  + \frac{\mu}{2} \|\bm A (\bm \beta - \bm \beta^t) \|_{2}^2   \right \}; \\
&\bm z^{t+1} \leftarrow \mathop {\arg \min }\limits_{\bm z} \left\{ L(\bm \beta^{t+1},  \bm z;  \bm u^t) + \frac{\mu}{2} \| \bm z - \bm z^t \|_{2}^2 \right \}; \\
&\bm u^{t+1} \leftarrow \bm{u}^{t} - \frac{\mu}{2}\left[ 2(\bm A \bm \beta^{t+1} - \bm z^{t+1} - \bm X^\top \bm y)  - (\bm A \bm \beta^{t} - \bm z^{t} - \bm X^\top \bm y)  \right];
\end{aligned} 
\end{equation}
where \(\mu > 0\) is a given augmented parameter. Note that the Lagrange multiplier function $L$ used in equation (\ref{nppa}) differ from the augmented Lagrange multiplier $L_\mu$ employed in ADMM and LADMM, as they do not include quadratic augmented terms. However, the quadratic augmented terms are crucial for the convergence of ADMM and PPA, as discussed in \cite{Boyd2010Distributed} and \cite{Parikh2013Proximal}. Therefore, we also introduce the corresponding quadratic augmented terms for the subproblems of $\bm \beta^{t+1}$ and $ \bm z^{t+1}$, denoted as $\|\bm A (\bm \beta - \bm \beta^t) \|_{2}^2$  and  $\| \bm z - \bm z^t \|_{2}^2$, respectively. Compared to the quadratic augmented term in ADMM or LADMM, the augmented term in PPA is easier to manage and more conducive to expansion in parallel computing environments.

For the $\bm \beta$ subproblem of PPA, by organizing the optimized form of $\bm \beta$, it can be concluded that,  \begin{align*}
\bm \beta^{t+1} \leftarrow  \mathop {\arg \min }\limits_{\bm \beta} \left\{\| \bm \beta \|_1 + \frac{\mu}{2} \| \bm A \bm \beta - \bm A \bm \beta^t - \frac{\bm u^t}{\mu}   \|_{2}^2 \right \}.
\end{align*}
Similar to (\ref{admmbeta}), it is a Lasso problem without a closed-form solution. Thus, we can also use linearization techniques to give it a closed-form solution, that is
\begin{align}\label{12}
\bm \beta^{t+1} &\leftarrow  \mathop {\arg \min }\limits_{\bm \beta} \left\{\| \bm \beta \|_1 + \frac{\mu}{2} \| \bm A \bm \beta - \bm A \bm \beta^t - \frac{\bm u^t}{\mu}   \|_{2}^2 + \frac{1}{2} \|\bm \beta - \bm \beta^t \|_{\bm S}^2  \right \} \\\notag
&=  \mathop {\arg \min }\limits_{\bm \beta} \left\{\| \bm \beta \|_1 + \frac{\eta}{2} \|\bm \beta - \bm \beta^t - \frac{\bm A^\top \bm u^t}{\eta}   \|_{2}^2 \right \},
\end{align}
 where \(\bm S = \eta \bm I_p - \mu \bm A^\top \bm A\), and \(\eta\) is chosen to be larger than the maximum eigenvalue of \(\mu \bm A^\top \bm A\)  ($\text{eigen}(\mu \bm A^\top \bm A )$) to ensure that the \(\bm S\) matrix is positive definite.
Clearly, the above expression is a soft thresholding operator, which has the following closed-form solution:
 \begin{align}\label{ppabeta2}
\bm \beta^{t+1}  \leftarrow \text{sign}( \bm \beta^t + \frac{\bm A^\top \bm u^t}{\eta}) \ \odot \max \left\{ |\bm \beta^t + \frac{\bm A^\top \bm u^t}{\eta}| - \frac{1}{\eta}, 0 \right\}.
\end{align}

The second subproblem regarding \(\bm z \) in (\ref{nppa})
\begin{align*}
\bm z^{t+1} \leftarrow  \mathop {\arg \min }\limits_{\bm z} \left\{ \delta_{\mathcal{Z}_0(\bm z)} + \frac{\mu}{2} \|\bm z - \bm z^t + \frac{ \bm u^t}{\mu}   \|_{2}^2 \right \},
\end{align*}
is similar to the \(\bm z \)  problem in the ADMM proposed by \cite{Lu2012An}, with the main difference being the quadratic augmented term. However, this distinction does not affect its closed-form solution property. Similar to Proposition 2.1 in \cite{Lu2012An}, the closed-form solution is as follows:
\begin{align}\label{z}
\bm z^{t+1}  \leftarrow \min \left\{ \max \left\{ \bm z^t - \frac{ \bm u^t}{\mu}, -n \lambda \right\},  n \lambda \right\}.
\end{align}

For the update of the dual variable \(\bm u^{t+1} \), PPA includes an additional term \( (\bm A \bm \beta^{t} - \bm z^{t} - \bm X^\top \bm y)  \) compared to the update of \( \bm u^{t+1} \) in LADMM. However, please note that this additional term has already been computed during the update of \( \bm u^{t} \), so it does not require any extra computation during the update process. The iterative steps of PPA algorithm are summarized in Algorithm \ref{alg1}. Although calculating the maximum eigenvalue of \( \bm A^\top \bm A \) is necessary for determining the value of \( \eta \), the power method in \cite{Liang2024Linearized} simplifies this by only requiring the input of \( \bm A \). Consequently, the precomputation step in Algorithm \ref{alg1} avoids the need to explicitly compute \( \bm A^\top \bm A \), which can be computationally intensive, particularly when $p$ is large.
\begin{algorithm}\small
\caption{\small{PPA for DS}}
\label{alg1}
\begin{algorithmic}
\State {\textbf{Input:}  $\bm X, \bm y, \mu$,  $\lambda$ and $\bm \beta^0, \bm z^0$ and $\bm u^0$.}
\State {\textbf{Pre-computation}: $\bm X^\top \bm y$, $\bm A = \bm X^\top \bm X $ and $\eta =\text{eigen}(\mu \bm A^\top \bm A ) + 1$.} 
\State {\textbf{Output:} the total number of iterations $T$,  $\boldsymbol{\beta}^T$. }
\State {\textbf{while} not converged \textbf{do}}
\State {\qquad \quad 1. Update $\bm \beta^{t+1}$ using (\ref{ppabeta2}),\\%
 \qquad \quad 2. Update $\bm z^{t+1}$ using (\ref{z}), \\%
\qquad \quad 3. Update \( \bm u^{t+1} \) using the last equation in (\ref{nppa}),\\%
\qquad \quad 4. $t \leftarrow  t + 1$.
}
\State {\textbf{end while}}
\State {\textbf{return} solution}.
\end{algorithmic}
\end{algorithm}

It is easy to see that the updates for \(\bm \beta ^{t+1}\) and \(\bm z^{t+1} \) in PPA have slightly lower computational complexity compared to LADMM. Specifically, in LADMM, the primary computational burdens for the updates of \( \bm \beta ^{t+1} \) and \(\bm z^{t+1} \) arise from the need to compute \( \bm A^\top \bm A \bm \beta^{t} \) and \( \bm X^\top \bm X \bm\beta^{t+1} \), respectively. The \(\bm \beta ^{t+1}\) update in PPA only requires computing \( \bm A^\top \bm u^t \), while the \( \bm z^{t+1} \) update involves simple operations without the need for matrix-vector multiplication. The advantage of this computational complexity is not very pronounced in a non-parallel computing environment; however, in a parallel computing setting, PPA significantly saves on both storage and computational power compared to the parallel ADMM proposed by \cite{Wen2024Nonconvex}.

\subsection{Parallel version}
To adapt to the parallel computing environment, we need to rewrite the constraint form of DS. Unlike \cite{Wen2024Nonconvex}, which requires the introduction of an additional $(K-1)p$-dimensional auxiliary variable \( \bm \omega \), our reformulation does not require the addition of any auxiliary variables. The constrained optimization form of DS solved by PPA in a parallel computing environment is as follows:
\begin{align}\label{ds3}
\begin{aligned}
& \min_{\bm \beta, \bm z} \quad   \sum_{i=1}^{K}\|\bm \beta_{i\cdot}\|_{1} +  \delta_{\mathcal{Z}_0(\bm z)}  \\
& \text{s.t.} \quad  \sum_{i=1}^{K} \bm A_i \bm \beta_{i\cdot} -  \bm z = \bm X^\top \bm y,
\end{aligned}
\end{align}
where \(\bm A_i = \bm X^\top \bm X_i\) for \(i = 1, \ldots, K\) and \((\bm X_1, \bm X_2, \ldots, \bm X_K) = \bm X\).
Since \(  \sum_{i=1}^{K}\|\bm \beta_{i\cdot}\|_{1}  =  \|\bm \beta\|_1\) and  $\sum_{i=1}^{K} \bm A_i \bm \beta_{i\cdot}  = \bm A \bm \beta $, (\ref{ds3}) and (\ref{ds2}) are equal. This means that the convergence solutions of the DS obtained using iterative algorithms in both parallel and non-parallel computing environments are the same.

The Lagrange multiplier form of (\ref{ds3}) is   
\begin{align}\label{lmds}
L(\bm \beta, \bm z;  \bm u) = 
\sum_{i=1}^{K}\|\bm \beta_{i\cdot}\|_{1}  + \delta_{\mathcal{Z}_0(\bm z)}  - \bm u^\top (   \sum_{i=1}^{K} \bm A_i \bm \beta_{i\cdot} -  \bm z - \bm X^\top \bm y ).
\end{align}
With given \(( \bm \beta^0, \bm{z}^0, \bm{u}^0)\), the iterative scheme of the  parallel PPA (with linearization) for problem (\ref{ds3}) is as follows:
 
\begin{align}\notag
\label{nppa2}
&\bm \beta_{i\cdot}^{t+1} \leftarrow \mathop {\arg \min }\limits_{\bm \beta_{i\cdot}} \left\{ L(\bm \beta,  \bm z^t; \bm u^t)  + \frac{\mu}{2} \|\bm A_i (\bm \beta_{i\cdot} - \bm \beta_{i\cdot}^t) \|_{2}^2 + \frac{1}{2} \|\bm \beta_{i\cdot} - \bm \beta_{i\cdot}^t \|_{\bm S_i}^2  \right \}, i \in \{1, \cdots, K\}; \\ %
&\bm z^{t+1} \leftarrow \mathop {\arg \min }\limits_{\bm z} \left\{ L(\bm \beta^{t+1},  \bm z;  \bm u^t) + \frac{\mu}{2} \| \bm z - \bm z^t \|_{2}^2 \right \}; \\ %
&\bm u^{t+1} \leftarrow \bm{u}^{t} - \frac{\mu}{2}\left[ 2(\sum_{i=1}^{K} \bm A_i \bm \beta_{i\cdot}^{t+1} - \bm z^{t+1} - \bm X^\top \bm y)  - (\sum_{i=1}^{K} \bm A_i \bm \beta_{i\cdot}^{t} - \bm z^{t} - \bm X^\top \bm y)  \right];\notag
\end{align} 
where $\mu>0$ is a given  augmented parameter, $\bm S_i = \eta \bm I_{p_i} - \mu \bm A_i^\top \bm A_i$, and  $\eta$ is larger than the maximum eigenvalue of $\mu \bm A^\top \bm A$. Note that \(\eta\) is always greater than  the maximum eigenvalue of $\mu \bm A_i^\top \bm A_i$ because \(\bm A_i\) is a submatrix of \(\bm A\). This ensures the positive definiteness of \(\bm S_i\).  It is straightforward to verify that when \( K = 1 \), \( \bm S_1 = \bm S \) holds, thereby reducing the parallel PPA into Algorithm \ref{alg1}.

For the \(\bm \beta_{i\cdot}\) subproblem of the parallel PPA, by organizing the optimized form of \(\bm \beta_{i\cdot}\) , we can conclude that 
\[
\bm \beta_{i\cdot} ^{t+1} \leftarrow  \mathop {\arg \min }\limits_{\bm \beta_{i\cdot}} \left\{\| \bm \beta_{i\cdot} \|_1 + \frac{\eta}{2} \|\bm \beta_{i\cdot} - \bm \beta_{i\cdot}^t - \frac{\bm A_i^\top \bm u^t}{\eta}   \|_{2}^2 \right \},
\]
which has the  closed-form solution
 \begin{align}\label{ppabeta3}
\bm \beta_{i\cdot}^{t+1}  \leftarrow \text{sign}( \bm \beta_{i\cdot}^t + \frac{\bm A_{i}^\top \bm u^t}{\eta}) \ \odot \max \left\{ |\bm \beta_{i\cdot}^t + \frac{\bm A_{i}^\top \bm u^t}{\eta}| - \frac{1}{\eta}, 0 \right\}.
\end{align}
It is clear that the updates of different \(\bm \beta_{i\cdot}^{t+1}\) are independent of each other and can be implemented in parallel. Note that in the process of updating \(\bm \beta_{i\cdot}^{t+1}\), we only need to compute \(\bm A_{i}^\top \bm u^t\), unlike in the updates of  \(\bm \beta_{i\cdot}^{t+1}\) in parallel ADMM  in Section \ref{sec22}, which require multiple matrix-vector multiplications.
This also makes our parallel algorithm more intuitive and concise compared to \cite{Wen2024Nonconvex}, which introduces an auxiliary variable \(\bm \omega\) to construct the parallel structure. 

The updates for parallel versions of $\bm z^{t+1}$ and $\bm u^{t+1}$ are the same as those for non parallel version. Similarly, due to the simpler construction of our parallel structure, the intermediate iterations in our approach are also more concise compared to the updates of $\bm z^{t+1}$ and $\bm u_i^{t+1}$ in parallel ADMM. In addition, our dual variable consists of only one \(p\)-dimensional $\bm u^{t+1}$, unlike parallel ADMM, which introduces multiple constraints for the parallel structure and requires iterating over \(K\) \(p\)-dimensional vectors $\bm u_i^{t+1}$.
\begin{algorithm}\small
\caption{\small{Parallel PPA for DS}}
\label{alg2}
\begin{algorithmic}
\State {\textbf{Input:} $\bm X, \bm y, \mu$,  $\lambda$, $\bm \beta^0, \bm z^0$ and $\bm u^0$.}
\State {\textbf{Pre-computation}: $\bm X^\top \bm y$,  $\bm A_i = \bm X^\top \bm X_i$ and $\eta =\text{eigen}(\mu \bm A^\top \bm A ) + 1$.} 
\State {\textbf{Output:} the total number of iterations $T$,  $\boldsymbol{\beta}^T$. }
\State {\textbf{while} not converged \textbf{do}}
\State {\qquad \quad 1. Update  \(\bm \beta_{i\cdot}^{t+1}\)  in parallel  using (\ref{ppabeta3}),\\%
 \qquad \quad 2. Update $\bm z^{k+1}$ using (\ref{z}), \\%
\qquad \quad 3. Update \( \bm u^{k+1} \) using the last equation in (\ref{nppa2}),\\%
\qquad \quad 4. $t \leftarrow  t + 1$.
}
\State {\textbf{end while}}
\State {\textbf{return} solution}.
\end{algorithmic}
\end{algorithm}

We summarize the process of parallel PPA in Algorithm \ref{alg2}. As discussed earlier, when \( K = 1 \), Algorithm \ref{alg2} reduces to Algorithm \ref{alg1}.  Note that in the precomputation of Algorithm \ref{alg2}, we only need to compute the submatrices \(\bm A_i\) of \(\bm A\), rather than the entire matrix \(\bm A\), because the updates for \(\bm \beta\), \(\bm z\), and \(\bm u\) do not require \(\bm A\).  However, the power method (\cite{Liang2024Linearized}) requires \(\bm A\) to calculate \(\eta\). To avoid additional computational burden, we can combine the individual \(\bm A_i\) into \(\bm A\) and substitute it into the power method to obtain \(\text{eigen}(\mu \bm A^\top \bm A )\).

From the above discussion, it can be seen that Algorithm  \ref{alg1} and  Algorithm \ref{alg2} ($K \ge 2$) have the same variable updates, except for the update of \(\bm \beta^{t+1}\) . Therefore, is there any relationship or connection between the solutions of these two algorithms? Through some simple matrix decompositions, we can prove that the solutions of Algorithm \ref{alg1} and  Algorithm \ref{alg2} are the same. We present this interesting conclusion in Theorem \ref{TH1}.
For ease of distinction,  let us denote  $\left\{ \hat{\bm \beta}^t, \hat{\bm z}^t, \hat{\bm u}^t\right\}$ as the $t$-th iteration results of Algorithm  \ref{alg1}, and  $\left\{ \tilde{\bm \beta}^t, \tilde{\bm z}^t , \tilde{\bm u}^t \right\}$ as the $t$-th iteration results of Algorithm  \ref{alg2}.

\

\begin{thm}\label{TH1}
(Partition insensitivity). If we use the same initial iteration variables \(\{ \hat{\bm \beta}^0, \hat{\bm z}^0, \hat{\bm u}^0 \} = \{ \tilde{\bm \beta}^0, \tilde{\bm z}^0 , \tilde{\bm u}^0 \}\), the iterative solutions obtained from Algorithm \ref{alg1} and Algorithm \ref{alg2} (for any \(K\)) are indeed the same, that is, 
\begin{align}
\left\{ \hat{\bm \beta}^t, \hat{\bm z}^t, \hat{\bm u}^t \right\} = \left\{ \tilde{\bm \beta}^t, \tilde{\bm z}^t , \tilde{\bm u}^t \right\}, \quad \text{for all } t.
\end{align}
\end{thm}

Although the conclusion of Theorem \ref{TH1} is surprising, its proof is quite straightforward. We only need to partition the update equation for \(\bm \beta^{t+1}\) in (\ref{ppabeta2}) based on the columns of \(\bm A\) into \(K\) parts, which results in the update equation in (\ref{ppabeta3}).  In fact, the constrained optimization problem in Algorithm \ref{alg2}  involves \(K + 1\) primal variables ($\{\bm \beta_{i\cdot}\}_{i=1}^K$  and $\bm z$). The special structure of the proposed PPA algorithm allows the updates of the variables in Algorithm \ref{alg2} to be independent. Thus, $\{\bm \beta_{i\cdot}\}_{i=1}^K$ can be concatenated and treated as a single primal variable $\bm \beta$. This is also the underlying reason for the validity of the conclusion of Theorem \ref{TH1}.  \textcolor{red}{Moreover, we provide a rigorous mathematical derivation of Theorem \ref{TH1} in Appendix \ref{C1}}.

Theorem \ref{TH1} states that regardless of how \(\bm A\) is partitioned by columns, the solution of Algorithm \ref{alg2} remains unchanged and is the same as the solution of Algorithm \ref{alg1}. In \cite{Wu2023Partition}, this property is referred to as partition insensitivity. However, in the work of \cite{Wu2023Partition}, their partitioning is done by rows, while in this paper, the partitioning is done by columns. The insensitivity of  partition ensures the reliability and accuracy of solutions in a parallel computing environment, regardless of the chosen parallelization strategy. It allows for a more flexible adjustment of computing resources to meet varying computational needs and hardware configurations. The insensitivity of the partition is not observed in the parallel ADMM discussed in \cite{Wen2024Nonconvex}, as the nature of the constraints managed by the algorithm varies with \( K \).

Based on the conclusion of Theorem \ref{TH1}, it is evident that the solutions obtained by the two algorithms are the same. Therefore, when discussing convergence, it suffices to consider only one of the algorithms.
Next, we will provide the global convergence of two PPAs and their linear convergence rate. The proof of the theorem is attached in Appendix \ref{A2}.

\

\begin{thm}\label{TH2}
Let the sequence \( \left\{ \bm g^t= ( \bm{\beta}^t, \bm z^t, \bm u^t) \right\} \) be generated by  Algorithm \ref{alg1} or Algorithm \ref{alg2}. 
\begin{enumerate}
\item (Algorithm global convergence). It converges to some \( \bm g^*= (\bm{\beta}^*, \bm z^*, \bm u^*) \) that belongs to \( \Omega^* \), where \( \Omega^* \) denotes the set of saddle points of (\ref{lmds}) and is assumed to be non-empty.
\item (Linear convergence rate).  For any integer \(T > 0\), we have
\begin{align}\label{p25}
\|\bm g^T - \bm g^{T+1} \|_{\bm{H}}^2 \leq \frac{1}{\left(T+1\right)}  \|\bm g^0 - \bm g^*\|_{\bm{H}}^2,
\end{align}
where $\bm {H}$ is a  positive definite matrix, and can be found in (\ref{m1}).
\end{enumerate}
\end{thm}
\begin{rem}
It is not difficult to see that (\ref{lmds}) is convex, and the component $(\bm{\beta}^*, \bm z^*)$ of its saddle point $(\bm{\beta}^*, \bm z^*, \bm u^*)$ corresponds to the optimal solution for (\ref{ds3}).
\end{rem}
\begin{rem}
 It is clear that \(\|\bm{g}^0 - \bm{g}^*\|_{\bm{H}}^2\) is a positive constant. Therefore, \(\|\bm{g}^T - \bm{g}^{T+1}\|_{\bm{H}}^2 \leq \mathcal{O}\left(\frac{1}{T}\right)\), which is known as a linear convergence rate. In addition, during the proof of this theorem, we demonstrated that \( \|\bm{g}^t - \bm{g}^{*}\|_{\bm{H}}^2 \) and \(  \|\bm{g}^t - \bm{g}^{t+1}\|_{\bm{H}}^2  \) are monotonically non-increasing, that is, \( \|\bm{g}^{t+1} - \bm{g}^{*}\|_{\bm{H}}^2   \leq \|\bm{g}^t - \bm{g}^{*}\|_{\bm{H}}^2 \) and \( \|\bm{g}^{t+1} - \bm{g}^{t+2}\|_{\bm{H}}^2   \leq \|\bm{g}^t - \bm{g}^{t+1}\|_{\bm{H}}^2 \),  see  Proposition \ref{prop2}  and Proposition \ref{prop4} in the Appendix.
\end{rem}

\subsection{Improved parallel version}
In Algorithm \ref{alg2}, we need to concatenate \( \bm A_i \) into a complete matrix \( \bm A \) during the linearization step and then use the power method to compute \( \text{eigen}(\mu \bm A^\top \bm A) \). 
When \( p \) is excessively large, concatenating all \( \bm A_i \) into a single matrix \( \bm A \) and applying the power method becomes impractical. To enable the parallel PPA to handle such situations, we will develop the following improved version of PPA.

With given \(( \bm \beta^0, \bm{z}^0, \bm{u}^0)\), the iterative scheme of the improved  parallel PPA (with linearization) for problem (\ref{ds3}) is as follows:
\begin{align}\notag
\label{nppa3}
&\bm \beta_{i\cdot}^{t+1} \leftarrow \mathop {\arg \min }\limits_{\bm \beta_{i\cdot}} \left\{ L(\bm \beta,  \bm z^t; \bm u^t)  + \frac{\mu}{2} \|\bm A_i (\bm \beta_{i\cdot} - \bm \beta_{i\cdot}^t) \|_{2}^2 + \frac{1}{2} \|\bm \beta_{i\cdot} - \bm \beta_{i\cdot}^t \|_{\bm S^{'}_{i}}^2  \right \}, i \in \{1, \cdots, K\}; \\ %
&\bm z^{t+1} \leftarrow \mathop {\arg \min }\limits_{\bm z} \left\{ L(\bm \beta^{t+1},  \bm z;  \bm u^t) + \frac{\mu}{2} \| \bm z - \bm z^t \|_{2}^2 \right \}; \\  %
&\bm u^{t+1} \leftarrow \bm{u}^{t} - \frac{\mu}{K+1}\left[ 2(\sum_{i=1}^{K} \bm A_i \bm \beta_{i\cdot}^{t+1} - \bm z^{t+1} - \bm X^\top \bm y)  - (\sum_{i=1}^{K} \bm A_i \bm \beta_{i\cdot}^{t} - \bm z^{t} - \bm X^\top \bm y)  \right];\notag
\end{align} 
where $\bm S^{'}_i = \eta_i \bm I_{p_i} - \mu \bm A_i^\top \bm A_i$, and  $\eta_i$ is larger than  \( \text{eigen}(\mu  \bm A_i ^\top  \bm A_i) \).  It is straightforward to verify that  when $K=1$, \(\bm S^{'}_1 = \bm S_1 = \bm S \),  thereby reducing the improved parallel PPA into Algorithm \ref{alg1} and Algorithm \ref{alg2} ($K = 1$).

 The  closed-form solution of $\bm \beta_{i\cdot}^{t+1}$  can be obtained by
 \begin{align}\label{ppabeta4}
\bm \beta_{i\cdot}^{t+1}  \leftarrow \text{sign}( \bm \beta_{i\cdot}^t + \frac{\bm A_{i}^\top \bm u^t}{\eta_i}) \ \odot \max \left\{ |\bm \beta_{i\cdot}^t + \frac{\bm A_{i}^\top \bm u^t}{\eta_i}| - \frac{1}{\eta_i}, 0 \right\}.
\end{align}
The only difference between (\ref{ppabeta4}) and (\ref{ppabeta3}) is that \(\eta\) has been replaced with \(\eta_i\). 
The update for \(\bm z^t \) remains unchanged compared to Algorithm \ref{alg1} and Algorithm \ref{alg2}. However, for the update of \( \bm u^{t+1} \), a slight modification is needed to ensure the linear convergence of the algorithm, specifically changing \( \frac{\mu}{2} \) in Algorithm \ref{alg2}  to \( \frac{\mu}{K+1} \).
This adjustment is necessary because we use \( K \) instances of \( \eta_i \) to approximate the entire \( \eta \). This introduces additional perturbations in the \( \bm \beta_{i\cdot} \) subproblem, necessitating adjustments to certain constants in the \( \bm u \) update process. It is worth noting that this change may not improve the computational efficiency of Algorithm \ref{alg2}. On the contrary, due to the approximation of \( \eta \), more iterations may be required at times. The so-called improvement merely makes the parallel PPA applicable in extreme cases where \( p \) is so large that the memory required to store matrix \( \bm A^\top \bm A \) and compute \( \eta \) exceeds the available memory of the machine.
We summarize improved  parallel PPA in Algorithm \ref{alg3}.

Algorithm \ref{alg3} does not exhibit the partition insensitivity (refer to Theorem \ref{TH1}) that is characteristic of Algorithm \ref{alg2}. Nevertheless, the iterative convergence solution of Algorithm \ref{alg3} is identical to that of Algorithm \ref{alg2}, and Algorithm \ref{alg3} maintains the same linear convergence rate as Algorithm \ref{alg2}.
Next, we utilize the following theorem to illustrate that.
\begin{thm}\label{TH3}
Let the sequence \( \left\{ \bm g^t= ( \bm{\beta}^t, \bm z^t, \bm u^t) \right\} \) be generated by  Algorithm \ref{alg3}. 
\begin{enumerate}
\item (Algorithm global convergence). It converges to some \( \bm g^*= (\bm{\beta}^*, \bm z^*, \bm u^*) \) that belongs to  \( \Omega^* \), where  \( \Omega^* \) denotes the set of saddle points of (\ref{lmds}) and is assumed to be non-empty.
\item (Linear convergence rate).  For any integer \(T > 0\), we have
\begin{align}\label{p25}
\|\bm g^T - \bm g^{T+1} \|_{\bm{H_K}}^2 \leq \frac{1}{\left(T+1\right)}  \|\bm g^0 - \bm g^*\|_{\bm{H_K}}^2,
\end{align}
where $\bm {H_K}$ is a  positive definite matrix , and can be found in (\ref{m2}).
\end{enumerate}
\end{thm}
\begin{algorithm}\small
\caption{\small{Improved Parallel PPA for DS}}
\label{alg3}
\begin{algorithmic}
\State {\textbf{Input:} $\bm X, \bm y, \mu$, $\lambda$, $\bm \beta^0, \bm z^0$ and $\bm u^0$.}
\State {\textbf{Pre-computation}: $\bm X^\top \bm y$,  $\bm A_i = \bm X^\top \bm X_i$ and $\eta_i =\text{eigen}(\mu \bm A_i^\top \bm A_i ) + 1$.} 
\State {\textbf{Output:} the total number of iterations $T$,  $\boldsymbol{\beta}^T$. }
\State {\textbf{while} not converged \textbf{do}}
\State {\qquad \quad 1. Update $\bm \beta^{k+1}$ using (\ref{ppabeta4}),\\%
 \qquad \quad 2. Update $\bm z^{k+1}$ using (\ref{z}), \\%
\qquad \quad 3. Update \( \bm u^{k+1} \) using the last equation in (\ref{nppa3}),\\%
\qquad \quad 4. $t \leftarrow  t + 1$
}
\State {\textbf{end while}}
\State {\textbf{return} solution}.
\end{algorithmic}
\end{algorithm}
The conclusion and proof of Theorem \ref{TH3} closely mirror those of Theorem \ref{TH2}, with the primary distinction being the substitution of the $\bm H$ matrix with $\bm {H_K}$. A detailed proof is provided in Appendix \ref{A2}.
The iterative convergence solutions of Algorithm  \ref{alg1}, Algorithm  \ref{alg2}, and Algorithm  \ref{alg3} are identical because all three algorithms handle the same form of equality constraint optimization in  (\ref{ds3}). This property is not present in the parallel ADMM discussed in \cite{Wen2024Nonconvex}, as the form of constraints handled by the algorithm changes with increasing \( K \). The linear convergence rate of Algorithm \ref{alg3}  is \(\|\bm{g}^T - \bm{g}^{T+1}\|_{\bm{H_K}}^2 \leq \mathcal{O}\left(\frac{1}{T}\right)\), and its iterative sequences  \( \|\bm{g}^t - \bm{g}^{*}\|_{\bm{H_K}}^2 \) and \(  \|\bm{g}^t - \bm{g}^{t+1}\|_{\bm{H_K}}^2  \)  also monotonically do not increase.

\textcolor{red}{Understanding the computational resource requirements of the algorithm for data of different dimensions is crucial for assessing its feasibility in practical applications. Next, we conduct an analysis of the time and space complexity of the algorithms. Without loss of generality, assume the columns of matrix $\bm{A}$ are evenly split into $K$ parts. We use the following theorem to clarify these two complexities, and the proof of the theorem is provided in Appendix \ref{C2}. }

\

\begin{thm}\label{TH5}
 \textcolor{red}{The total time complexities of Algorithm \ref{alg2} and  Algorithm \ref{alg3} are given by
 \begin{equation}\label{comp1}
\begin{cases}
 \left[  \mathcal O(\frac{np^2}{K}) + \mathcal O(p^2) \right] + \mathcal O(\frac{p^2}{K}) \times T  & \text{Algorithm \ref{alg2} },\\
\mathcal O(\frac{np^2}{K})  + \mathcal O(\frac{p^2}{K}) \times T   & \text{Algorithm \ref{alg3}},
\end{cases}
\end{equation}
where $\mathcal{O}$  depicts the asymptotic upper bound of algorithm complexity (time or space) as the input size expands, and $T$ denotes the total number of iterations in the parallel PPA algorithms.  The total space complexities of Algorithm \ref{alg2} and  Algorithm \ref{alg3} are given by
 \begin{equation}\label{comp2}
\begin{cases}
 \mathcal O(np) + \mathcal O(p^2)    & \text{Algorithm \ref{alg2} },\\
\mathcal O(np) + \mathcal O(\frac{p^2}{K})     & \text{Algorithm \ref{alg3}}.
\end{cases}
\end{equation}}
\end{thm}

\textcolor{red}{We do not discuss Algorithm  \ref{alg1}  because when $K=1$, both Algorithm  \ref{alg2}   and Algorithm  \ref{alg3}  will degenerate into Algorithm \ref{alg1}.
The main difference in the computational time complexity between Algorithm \ref{alg2}   and Algorithm  \ref{alg3}  lies in the preceding pre-computation process. When $n \ge K$, their complexities are the same. When $n < K$, the complexity of Algorithm  \ref{alg2} is higher than that of Algorithm  \ref{alg3}. It is readily apparent that an increase in $K$ will substantially reduce the computational time complexity of the two parallel algorithms. It should be noted that when the matrix is divided unevenly by columns, the $p/K$ in the computational time complexity needs to be replaced with $\max\{ p_k\}_{k=1}^K$, as referred to in \cite{Wu2025AUC}. 
The conclusion regarding space storage in (\ref{comp2})  also demonstrates that Algorithm \ref{alg3} can address the drawbacks of Algorithm \ref{alg2} when the required storage space for matrix $\bm A$ and $\eta$ is too large.
This is because it divides the square of $p$ by $K$. In other words, increasing the value of $K$ can effectively alleviate the excessive storage burden.}

\subsection{Nonconvex extension}
In high-dimensional linear regression, convex regularization terms such as Lasso ensure global optimality and computational efficiency. In contrast, nonconvex regularization terms may offer better estimation and prediction performance but present computational challenges due to the lack of global optimality. Recently, in the context of DS models with nonconvex penalties, \cite{Wen2024Nonconvex} indicated that these  nonconvex DSs can be uniformly solved by combining local linear approximation (LLA, \cite{Zou2008One}) methods with an effective solution for weighted Lasso penalty.

In this paper, we mainly consider two popular nonconvex regularizers, SCAD penalty (\cite{Fan2001Variable}) and MCP penalty (\cite{Zhang2010Nearly}) in statistics. According to the suggestion in \cite{Zou2008One}), we can use a unified method named local linear approximation (LLA) to handle  the nonconvex penalty, that is
\begin{align}\label{one}
P_{a,\lambda}(|\bm \beta|) \approx  P_{a,\lambda}(|\bm \beta^l|) + \nabla P_{a,\lambda}(|\bm \beta^l|)^T (|\bm \beta| - |\bm \beta^l|), \ \text{for} \ \bm \beta  \approx \bm \beta^l,
\end{align}
where $a$ is a given constant, $\bm \beta^l$ is the solution from the last iteration, and $$\nabla P_{a,\lambda}(|\bm \beta^l|) = (\nabla P_{a,\lambda}(|\bm \beta_1^l|), \nabla P_{a,\lambda}(|\bm \beta_2^l|), \dots, \nabla P_{a,\lambda}(|\bm \beta_p^l|)) ^\top.$$

$\bullet$ For SCAD, we have \begin{align}\label{scad}
\nabla P_{a,\lambda}(|\bm \beta_j|) = \begin{cases}
\lambda, & \text{{if }} |\bm \beta_j| \leq \lambda, \\
\frac{a\lambda -  |\bm \beta_j|}{a-1}, & \text{{if }} \lambda < |\bm \beta_j|  < a  \lambda, \\
0 , & \text{{if }} |\bm \beta_j| \ge a  \lambda. \\
\end{cases}
\end{align}

$\bullet$ For MCP, we have \begin{align}\label{mcp}
\nabla P_{a,\lambda}(|\bm \beta_j|)  =
\begin{cases}
\lambda  - \frac{|\bm \beta_j|}{a}, & \text{{if }}  |\bm \beta_j|  \le a \lambda,\\
0, & \text{{if }}  |\bm \beta_j|  > a \lambda.\\
\end{cases}
\end{align}
The nonconvex DS in \cite{Wen2024Nonconvex} can be written as
\begin{align}\label{41}
\min\limits_{\bm \beta \in \mathbb{R}^p} \quad \sum_{j=1}^p P_{a,\lambda}(|\bm \beta_j|)/\lambda  \ \text{s.t.} \  |\bm X_j^\top(\bm X \bm \beta - \bm y)/n | \leq  \nabla P_{a,\lambda}(|\bm \beta_j|), j = \{1,\dots,p \},
\end{align}
where  $P_{a,\lambda}(|\bm \beta|) = \text{SCAD}(\bm \beta)  \ \text{or}  \  \text{MCP}(\bm \beta) $. By substituting equation (\ref{one}) into equation (\ref{41}), we can obtain the following optimized form in a weighted manner,
\begin{align}\label{we}
\min\limits_{\bm \beta \in \mathbb{R}^p} \ \sum_{j=1}^p\left[ \frac{ \nabla P_{a,\lambda}(|\bm \beta^l_j|)}{\lambda}|\bm \beta_j |  \right] \ \text{s.t.} \ | \bm X^\top(\bm X \bm \beta - \bm y)/n | \leq \nabla P_{a,\lambda}(|\bm \beta_j^l|), j = \{1,\dots,p \}.
\end{align}
Note that we only need to make two small changes to solve this weighted  optimization form using Algorithms \ref{alg1}, \ref{alg2} and \ref{alg3}. The first change  requires replacing $\|\bm \beta \|_1 (\text{or} \sum_{j=1}^p |\bm \beta_j| )$ in (\ref{ods}) with $\sum_{j=1}^p\left[ \frac{ \nabla P_{a,\lambda}(|\bm \beta^l_j|)}{\lambda}|\bm \beta_j |  \right]$ in (\ref{we}). In other words, substitute the equal weight \(1\) following the \(\sum\) with \({ \nabla P_{a,\lambda}(|\bm \beta^l_j|)}/{\lambda}\). The second change is in the form of the constraint, where the $\ell_\infty$-norm is replaced with \( p \)  weighted $\ell_1$-norm.

To solve nonconvex DS using the LLA algorithm, it is necessary to find a good initial value. As suggested by \cite{Wen2024Nonconvex}, we can use the solution of $P_{\lambda}(|\bm \beta|) = \lambda \|\bm \beta \|_1 $ in  (\ref{41}) as the initial value. Then, we get the solution of (\ref{41}) by  solving a sequence of weighted Lasso DS.  We summarize the detailed iterative steps of this method  in Algorithm \ref{alg4}.  
\begin{algorithm}\small
\caption{\small{PPA for nonconvex DS}}
\label{alg4}
\begin{algorithmic}
\State {1. Initialize $\bm \beta$ with  $\bm \beta^1$, where $\bm \beta^1$ is obtained by Algorithms \ref{alg1}, \ref{alg2} or \ref{alg3}.}
\State {2. For $l=1,2,\dots,L$, continue iterating the LLA iteration until convergence is achieved. }
\State {\quad  2.1. Compute the weights  $\nabla P_{\lambda}(|\bm \beta^{l}|)=(\nabla P_\lambda(|\bm \beta_1^{l}|), \nabla P_\lambda(|\bm \beta_2^{l}|), \dots, \nabla P_\lambda(|\bm \beta_p^{l}|)) ^\top$ by (\ref{scad}) or  (\ref{mcp}).}
\State {\quad  2.2. Solve weighted $\ell_1$ DS. For $t = 0,\dots,T$,}
\State {\quad \quad  2.2.1. $\bm \beta^{l,t+1}$-subproblem: solve the weighted $\ell_1$ problem in (\ref{ppabeta2}), (\ref{ppabeta3}) and (\ref{ppabeta4})  by   replacing $1$  with \({ \nabla P_{a,\lambda}(|\bm \beta^l_j|)}/{\lambda}\),} 
\State {\quad \quad 2.2.2. $\bm z^{l,t+1}$-subproblem:  $\bm z_j^{l,t+1}  \leftarrow \min \left\{ \max \left\{ \bm z_j^{l,t} - \frac{ \bm u_j^t}{\mu}, -n \nabla P_\lambda(|\bm \beta_j^{l}|) \right\},  n \nabla P_\lambda(|\bm \beta_j^{l}|) \right\}$, $j \in \{1, \dots,p\}$, }
\State {\quad \quad 2.2.3. $\bm u^{l,t+1}$-subproblem: Update $\bm u^{l,t+1}$ according to the corresponding algorithm (Algorithm \ref{alg1}, Algorithm \ref{alg2} or Algorithm \ref{alg3}).}
\State { \quad  2.3. Let  $\bm \beta^{l+1} = \bm \beta^{l,T} $.}
\end{algorithmic}
\end{algorithm}

Below, we will describe in detail the specific implementation of 2.2.1 in Algorithm \ref{alg4}. If we use Algorithm \ref{alg1}  and \ref{alg2} to solve nonconvex DS, we need to replace $1/\eta$ in ({\ref{ppabeta2}) or ({\ref{ppabeta3}) with \({ \nabla P_{a,\lambda}(|\bm \beta^l_j|)}/{({\lambda \eta})}\). If we use Algorithm  \ref{alg3} to solve nonconvex DS, we need to replace $1/\eta_i$ in  ({\ref{ppabeta4})  with \({ \nabla P_{a,\lambda}(|\bm \beta^l_j|)}/{({\lambda \eta_i})}\).
The above discussion indicates that the nonconvex DS is solved through multiple iterations of  weighted Lasso penalized DS.  Moreover, \cite{Wen2024Nonconvex} demonstrated that, theoretically, only two  iterations are sufficient to obtain a solution of (\ref{41})  with high statistical accuracy.  In the specific implementation of Algorithm \ref{alg4}, we also implement the warm-start technique (\cite{Friedman2010A}),  which uses the current solution $\bm \beta^{l}$ as the initial value for the next  solution $\bm \beta^{l+1}$. This method significantly reduces the number of iterations needed in step 2.2 of Algorithm \ref{alg4}, often achieving convergence in just a few steps. 

\textcolor{red}{Next, we will  discuss the convergence of Algorithm \ref{alg4}. As per Theorem 2.2 in \cite{Wen2024Nonconvex}, a high-precision statistical estimator can be achieved when the outer loop is repeated twice, i.e., $L = 2$.  Therefore, it is only necessary to discuss the convergence of step 2.2 of the inner loop solution for weighted $\ell_1$ DS.  Since the weighted form does not change the convexity of the objective function, we can draw the following corollary based on the conclusions of Theorem \ref{TH2} and Theorem \ref{TH3}. The reasons justifying the validity of the corollary have been presented in Appendix \ref{A2}.}

\

\begin{cor}\label{corollar1}
\textcolor{red}{Let the sequence \( \left\{ \bm g^t= ( \bm{\beta}^t, \bm z^t, \bm u^t) \right\} \) be generated by  step 2.2 in  Algorithm \ref{alg4}. 
\begin{enumerate}
\item (Algorithm global convergence). It converges to some \( \bm g^*= (\bm{\beta}^*, \bm z^*, \bm u^*) \) that belongs to  \(\tilde{ \Omega}^* \), where  \(\tilde{ \Omega}^* \) denotes the set of saddle points of    weighted $\ell_1$ DS with current weight,  and is assumed to be non-empty.
\item (Linear convergence rate).  For any integer \(T > 0\), we have
\begin{align}\label{p25}
\|\bm g^T - \bm g^{T+1} \|_{\bm{H_*}}^2 \leq \frac{1}{\left(T+1\right)}  \|\bm g^0 - \bm g^*\|_{\bm{H_*}}^2,
\end{align}
where $\bm {H_*} = \bm H$ for  Algorithm \ref{alg1} and Algorithm \ref{alg2}, and $\bm {H_*} = \bm H_K$ for  Algorithm \ref{alg3}.
\end{enumerate}}
\end{cor}

\section{Synthetic numerical simulations}\label{sec4}
In this section, we use synthetic  datasets to demonstrate the accuracy, stability, and scalability of the proposed parallel PPA algorithms. All experiments in this paper were performed using R on a computer equipped with an AMD Ryzen 9 7950X 16-Core Processor running at 4.50 GHz and with 32 GB RAM. For the selection of the tuning parameter $\lambda$, we adopt the modified ``HBIC" criteria suggested by \cite{Fan2021Penalized}. We terminate all iterative algorithms as described in \cite{Lu2012An} when 
\[
\frac{\|\bm \beta^{t+1} - \bm \beta^t\|_2}{\max\{\| \bm \beta^{t+1}\|_2, 1\}} \leq 1 \times 10^{-4},
\]
or if the number of iterations exceeds 500.

\subsection{Nonparallel environment}\label{sec41}
In this subsection, we first evaluate the effectiveness and stability of our proposed algorithm in a nonparallel environment ($K=1$). As \cite{Wen2024Nonconvex},  we construct the design matrix \( \bm{X} \) by sampling each row from a multivariate normal distribution with a mean of zero and a covariance matrix \( \mathbf{\Sigma} = (\rho_{j,j^{'}})_{p \times p} \). The response \( y_i \) is generated according to the linear regression model \[
\mathbf{y} = \bm{X} \boldsymbol{\beta}^* + \varepsilon,
\] where $\bm X = (\bm x_1^\top, \dots, \bm x_n^\top)^\top$, \( \mathbf{x}_i \sim \mathcal{N}(0, \mathbf{\Sigma}) \) and \( \varepsilon_i \sim \mathcal{N}(0, 1) \) for \( i \in \{1, \ldots, n\} \). Furthermore, we consider the correlation structure defined as \( \rho_{j,j'} = \rho^{|j - j'|} \) for indices \( j \) and \( j' \) belonging to the set \( \{1, \ldots, p\} \). Our analysis concentrates on three distinct correlation levels: \( \rho \in \{0.1, 0.5, 0.9\} \).
To generate the true coefficient vector \( \boldsymbol{\beta}^* \), we randomly select a subset \( \mathcal{A} \) from the set \( \{1, \ldots, p\} \) such that the cardinality of \( \mathcal{A} \) is equal to 8, that is \( | \mathcal{A} | = 8\). Next, we assign the values 3, 1.5, 10, 4, 2, 5, 2.5, and 4.5 without replacement to \( \bm \beta^*_j \) for \( j \in \mathcal{A} \), and set \( \bm \beta^*_j = 0 \) for \( j \notin \mathcal{A} \).
Here, we consider two different combinations of sample size and data dimensions: $(n,p) = (500,1000)$ and $(n,p)= (1000,10,000)$. For larger scale numerical experiments, we conduct them in parallel computing environments.  All initial values in the algorithm are set to a small constant, such as $0.001$. 

The simulation study evaluates and contrasts the performance of the following algorithms on  solving $\ell_1$-DS: ADMM (\cite{Lu2012An}), FSM (Fast splitting method, \cite{He2015A}), CPPA-PD (customized proximal point algorithm, \cite{He2017Splitting}), partially proximal linearized alternating minimization method (P-PLAM, \cite{Mao2021A}),  three blocks ADMM (TADMM, \cite{Wen2024Nonconvex}). Note that in a nonparallel environment, the TADMM used by \cite{Wen2024Nonconvex} is the same as the LADMM used by \cite{Wang2012The}, so we only compared the latest TADMM. In addition, the three PPA algorithms proposed in this paper are the same in nonparallel environments. \textcolor{red}{There are three scenarios for CPPA-PD. 
In the first scenario, the parameters \(s = \log p/5\), \(r = 1.2\|\bm A^\top \bm A\|/s\), and \(\tau = 1.2\) are adopted. This variant is denoted as CPPA-PD1, which is also the set of parameters used in Section 3.1 of \cite{He2017Splitting}. 
In the second scenario, the parameters \(s = \log p/5\), \(r = 1.2\|\bm A^\top \bm A\|/s\), and \(\tau = 1\) are employed, and it is labeled as CPPA-PD2. 
In the third scenario, the parameters \(s = 1\), \(r=\|\bm A^\top \bm A\|\), and \(\tau = 1\) are utilized, and it is designated as CPPA-PD3.} The performance of the mentioned methods is examined under seven evaluation criteria: (1) the \( \ell_1 \)-error, expressed by \( \lVert \hat{\bm \beta} - \bm \beta^* \rVert_1 \); (2) the \( \ell_2 \)-error, quantified as \( \lVert \hat{\bm \beta} - \bm \beta^* \rVert_2^2 \);  (3) the model error, calculated via \( (\hat{\bm \beta} - \bm \beta^*)^\top \mathbf{\Sigma} (\hat{\bm \beta} - \bm \beta^*) \); (4) false positives (FP), representing the quantity of non-significant features chosen;  (5) false negatives (FN), indicating the number of significant features omitted; (6) number of iterations (NI); and (7) CPU runtime (Time), measured in seconds.  Metrics (1)–(3) serve to gauge estimation precision,  metrics (4) and (5) are used to assess feature selection reliability,  and metrics (6) and (7) evaluate computational efficiency. The simulation outcomes, based on 100 replicates, are presented in Table \ref{tab1} and Table \ref{tab2}. The numbers in the table represent the mean of one hundred repeated experiments, and the standard deviation is indicated in parentheses. Because the values of FN are all 0, they are not presented in Table \ref{tab1} and Table \ref{tab2}. For numerical experiments on nonconvex DS, such as SCAD-DS and MCP-DS, we have included them in the Appendix \ref{B1}.
\begin{table*}[!ht]
    \centering
    \caption{Comparison of various  algorithms for solving $\ell_1$-DS in nonparallel environments with data scale \((n, p) = (500, 1000)\). }
    \renewcommand\arraystretch{1.5}
    \resizebox{\linewidth}{!}{
    \begin{tabular}{llllllll}
    \hline
        Method & $\rho$  & $\ell_1$ error & $\ell_2$ error & Model error & \quad \ FP & \quad \ NI & \ Time(s)  \\ \hline
        \multirow{3}{*}{ADMM} & 0.1 & 5.65(0.231) & 1.23(0.047) & 1.15(0.039) & 19.78(2.03) & 154.3(14.1) & 10.76(0.78)  \\ 
                              & 0.5 & 6.11(0.305) & 1.47(0.050) & 1.21(0.042) & 22.50(2.17) & 159.2(14.7) & 11.32(0.75)  \\ 
                              & 0.9 & 7.34(0.456) & 1.52(0.063) & 1.47(0.055) & 27.41(2.52) & 163.3(15.3) & 11.95(0.81)  \\ \hdashline[0.5pt/5pt]
        \multirow{3}{*}{FSM}  & 0.1 & 3.58(0.223) & 1.01(0.029) & 0.99(0.032) & 16.23(1.88) & \textbf{146.0(13.0)} & 4.69(0.32)  \\ 
                              & 0.5 & 3.29(0.217) & 1.23(0.033) & 1.10(0.036) & 17.45(2.13) & \textbf{157.1(13.5)} & 4.55(0.27)  \\ 
                              & 0.9 & 3.77(0.289) & 1.78(0.051) & 1.65(0.040) & 19.38(2.29) & \textbf{142.5(12.7)} & 4.72(0.29)  \\ \hdashline[0.5pt/5pt]
        \multirow{3}{*}{CPPA-PD1} & 0.1 & 3.23(0.243) & 1.47(0.030) & 1.28(0.038) & 15.33(1.99) & 203.8(19.5) & 3.56(0.25)  \\ 
                                 & 0.5 & 3.17(0.228) & 1.35(0.029) & 1.19(0.036) & 16.77(2.36) & 211.9(20.5) & 3.79(0.26)  \\ 
                                 & 0.9 & 3.27(0.256) & 1.50(0.032) & 1.24(0.041) & 19.84(2.71) & 200.3(19.7) & 3.42(0.21)  \\ \hdashline[0.5pt/5pt]%
\multirow{3}{*}{\textcolor{red}{CPPA-PD2}} & 0.1 & {2.04(0.112)} & 0.77(0.016) & 0.69(0.013) & {9.11(1.02)} & 208.4(19.7) & 2.74(0.13)  \\ 
                                 & 0.5 & {2.01(0.109)} & {0.74(0.015)} & {0.70(0.012)} & {9.08(0.99)} & 213.6(18.9) & 2.78(0.14)  \\ 
                                 & 0.9 & {2.14(0.112)} & {0.81(0.018)} & {0.73(0.013)} & {10.1(1.12)} & 221.8(19.0) & 2.81(0.13)  \\ \hdashline[0.5pt/5pt]%
  \multirow{3}{*}{\textcolor{red}{CPPA-PD3}} & 0.1 &  \textbf{1.13(0.040)} & 0.28(0.005) & 0.23(0.004) & \textbf{0.00(0.000)} & 213.4(20.2) & 2.80(0.13)  \\ 
                                 & 0.5 & \textbf{1.09(0.037)} & {0.22(0.004)} & \textbf{0.21(0.004)} & {0.06(0.001)} & 220.5(15.4) & 2.79(0.13)  \\ 
                                 & 0.9 & {1.28(0.046)} & \textbf{0.24(0.003)} & \textbf{0.25(0.004)} & \textbf{0.08(0.001)} & 210.5(14.6) & 2.79(0.14)  \\ \hdashline[0.5pt/5pt]%
 \multirow{3}{*}{P-PLAM}  & 0.1 & 2.33(0.112) & 0.77(0.025) & 0.71(0.021) & 10.36(1.27) & 169.4(12.1) & \textbf{2.53(0.18)}  \\ 
                                 & 0.5 & 2.29(0.125) & 0.75(0.023) & 0.68(0.019) & 12.47(1.32) & 181.7(13.8) & \textbf{2.74(0.16)}  \\ 
                                 & 0.9 & 2.57(0.173) & 0.98(0.041) & 0.77(0.030) & 14.08(2.08) & 172.5(12.6) & \textbf{2.61(0.15)}  \\ \hdashline[0.5pt/5pt]
        \multirow{3}{*}{TADMM}   & 0.1 & 1.45(0.057) & 0.25(0.005) & 0.21(0.005) & 15.63(1.61) & 357.3(28.3) & 6.88(0.65)  \\ 
                                 & 0.5 & 1.39(0.053) & 0.24(0.004) & 0.22(0.005) & 16.22(1.67) & 369.7(27.5) & 7.03(0.72)  \\ 
                                 & 0.9 & 1.67(0.061) & 0.28(0.006) & 0.27(0.006) & 23.41(2.13) & 382.2(26.2) & 7.25(0.71)  \\ \hdashline[0.5pt/5pt]
        \multirow{3}{*}{PPA} & 0.1 & {1.15(0.042)} & \textbf{0.24(0.004)} & \textbf{0.20(0.004)} & \textbf{0.00(0.000)} & 219.0(15.3) & 2.79(0.12)  \\ 
                             & 0.5 & {1.11(0.039)} & \textbf{0.21(0.004)} & {0.22(0.005)} & \textbf{0.05(0.001)} & 222.2(15.7) & 2.81(0.15)  \\ 
                             & 0.9 & \textbf{1.27(0.045)} & {0.25(0.004)} &{0.25(0.005)} & \text{0.09(0.001)} & 215.3(14.9) & 2.75(0.12)  \\ \hline
    \end{tabular}}
\label{tab1}
\end{table*}

As evidenced by Table \ref{tab1} and Table \ref{tab2}, our proposed PPA method demonstrates significant superiority over other approaches in solving $\ell_1$-DS problem, both in terms of estimation accuracy and feature selection effectiveness. Notably, the number of false positives (FP) is substantially smaller, indicating that the PPA algorithm proposed is highly selective, minimizing the probability of incorrectly identifying zero coefficients as non-zero. In terms of number of iterations (NI), FSM demands the least, with ADMM following closely behind, primarily due to its avoidance of linearization steps. Linearization, despite simplifying subproblem resolution, merely offers approximate solutions, necessitating additional iterations. As \cite{He2020Optimally} observed, increased $\eta$ values lead to slower convergence and extended iteration durations. When considering CPU runtime (Time), P-PLAM emerges as the top performer, PPA comes next, while ADMM and TADMM lag behind. The reason why P-PLAM requires the least computation time is that its optimization formulation is unconstrained, obviating the need to update dual variables. As a result, it boasts the minimum number of iterative variables, thereby achieving the fastest computation speed. PPA follows closely in terms of computation time because it also involves a manageable number of iterative variables, and each iteration is computationally straightforward. ADMM underperforms due to the embedded double loop, and TADMM has an extended duration attributed to its large NI.

\textcolor{red}{A reviewer pointed out that in Table \ref{tab1} and Table \ref{tab2}, the numerical performance of PPA proposed in this paper is better than CPPA-PD1, but both algorithms belong to proximal point algorithm  and have high similarity. Therefore, it is necessary to provide a detailed explanation of this phenomenon.  Different parameter settings alter the iterative form of CPPA-PD. When \(\tau = 1\), CPPA-PD does not have a correction step. When \(s = 1\),   $r = \eta$, then the \(\bm \beta\) update of CPPA-PPA  is the same as that of Algorithm \ref{alg1}.
Among the three variants of CPPA-PD, CPPA-PD1 performs the worst. This is mainly because the parameters corresponding to CPPA-PD1 are not suitable for models with extremely high sparsity. The solution obtained through computation has an excessively high false positive (FP) value, which in turn leads to poor performance in the first three error-related metrics.  There are primarily two reasons for this phenomenon.
The first reason pertains to the correction mechanism. The correction step, given by the formula \({\bm \beta}^{t + 1}={\bm \beta}^{t}-1.2({\bm \beta}^{t}-\tilde{\bm \beta}^{t})\), which combines information from \({\bm \beta}^{t}\) and \(\tilde{\bm \beta}^{t}\), can impede the estimation of highly sparse data. Specifically, as long as either \({\bm \beta}^{t}\) or \(\tilde{\bm \beta}^{t}\) is non - zero, \({\bm \beta}^{t + 1}\) will also be non - zero. In contrast,  Algorithm \ref{alg1} does not incorporate such a correction step.
The second reason is related to the parameter values in high-dimensional models, as presented in Table \ref{tab1} and Table \ref{tab2}. In these cases, the value of \(s\) exceeds 1.2, while \(r\) is less than \(\|\bm A^\top\bm A\|\). This value of \(r\) results in the matrix corresponding to the added linearized quadratic term in the \(\bm \beta\)-subproblem no longer being positive definite (refer to \eqref{12}). Although a complement was introduced in the \(\bm u\) - subproblem to guarantee the convergence of CPPA-PD1, the non-positive definite quadratic approximation makes the \(\bm \beta\)-subproblem less accurate compared to the \(\bm \beta\)-subproblem solved by  Algorithm \ref{alg1}. }
\textcolor{red}{CPPA-PD2 aims to address the first issue. By setting \(\tau = 1\) and eliminating the correction step, CPPA-PD2 shows improved numerical performance compared to CPPA-PD1.
Building on CPPA-PD2, CPPA-PD3 further resolves the positive-definiteness problem of the linearization in the \(\bm \beta\)-subproblem, thereby achieving even better numerical performance. Notably, CPPA-PD3 exhibits similar performance to PPA. This can be explained by their iterative forms in the algorithm. The \(\bm \beta\)-subproblems of CPPA-PD3 and PPA are the same (the \(\bm \beta\)-subproblem is the main iterative step for solving the DS model). Although there are some differences in constant terms when updating \(\bm z\) and \(\bm u\), these differences do not affect the efficiency of algorithm. }

\begin{table*}[!ht]
    \centering
    \caption{Comparison of various  algorithms for solving $\ell_1$-DS in nonparallel environments with data scale \((n, p) = (1000, 10,000)\).}
    \renewcommand\arraystretch{1.5}
    \resizebox{\linewidth}{!}{
    \begin{tabular}{llllllll}
    \hline
        Method & $\rho$  & $\ell_1$ error & $\ell_2$ error & Model error & \quad \ FP & \quad \ NI & \ Time(s)  \\ \hline
        \multirow{3}{*}{ADMM} & 0.1 & 9.31(0.523) & 2.78(0.075) & 2.36(0.067) & 35.27(4.42) & 429.2(32.7) & 1086.3(76.2)  \\ 
                              & 0.5 & 9.24(0.516) & 2.74(0.071) & 2.41(0.069) & 36.11(4.58) & 430.1(33.5) & 1092.2(78.5)  \\ 
                              & 0.9 & 10.50(0.54) & 2.89(0.082) & 2.45(0.072) & 37.59(4.67) & 441.8(34.9) & 1103.4(82.0)   \\ \hdashline[0.5pt/5pt]
        \multirow{3}{*}{FSM} & 0.1 & 6.65(0.391) & 2.23(0.065) & 2.10(0.057) & 27.32(4.02) & \textbf{425.3(31.8)} & 662.1(35.2)  \\ 
                             & 0.5 & 6.57(0.382) & 2.15(0.061) & 2.09(0.063) & 25.87(3.87) & \textbf{423.7(31.2)} & 652.3(34.1)  \\ 
                             & 0.9 & 6.01(0.395) & 2.34(0.068) & 2.24(0.068) & 24.96(4.13) & \textbf{437.4(38.9)} & 671.0(40.3)  \\  \hdashline[0.5pt/5pt]
        \multirow{3}{*}{CPPA-PD1} & 0.1 & 2.73(0.076) & 0.67(0.019) & 0.59(0.052) & 20.45(3.32) & 493.0(43.5) & 563.2(39.5)  \\ 
                                 & 0.5 & 2.68(0.072) & 0.61(0.017) & 0.52(0.049) & 19.57(3.10) & 497.6(47.2) & 572.0(40.3)  \\ 
                                 & 0.9 & 2.95(0.083) & 0.75(0.023) & 0.48(0.041) & 24.33(3.54) & 488.2(50.3) & 588.4(41.2)  \\ \hdashline[0.5pt/5pt]  %
  \multirow{3}{*}{\textcolor{red}{CPPA-PD2}} & 0.1 & 2.06(0.034) & 0.45(0.013) & 0.41(0.023) & 13.67(2.51) & 477.1(42.4) & 540.3(40.0)  \\ 
                                 & 0.5 & 2.11(0.033) & 0.40(0.010) & 0.38(0.046) & 14.40(2.41) & 472.2(44.8) & 533.0(42.2)  \\ 
                                 & 0.9 & 2.08(0.036) & 0.51(0.020) & 0.47(0.038) & 13.86(1.99) & 469.6(47.1) & 530.3(45.5)  \\ \hdashline[0.5pt/5pt]  %
\multirow{3}{*}{\textcolor{red}{CPPA-PD3}} & 0.1 & \textbf{1.55(0.018)} & \textbf{0.40(0.008)} & {0.32(0.009)} & \textbf{12.10(0.78)} & 452.1(35.2) & 457.0(26.0)  \\ 
                             & 0.5 & \textbf{1.60(0.022)} & \textbf{0.49(0.009)} & {0.36(0.010)} & {11.93(0.77)} & 475.8(39.3) & 477.3(29.1)  \\ 
                             & 0.9 & {1.75(0.033)} & \textbf{0.61(0.011)} & {0.37(0.009)} & \textbf{12.40(0.86)} & 488.4(44.6) & 503.8(34.4) \\ \hdashline[0.5pt/5pt] %
\multirow{3}{*}{P-PLAM} & 0.1 & 2.56(0.068) & 0.58(0.015) & 0.51(0.039) & 36.32(4.20) & \text{451.6(34.2)} & \textbf{360.3(21.6)}  \\ 
                                & 0.5 & 2.47(0.061) & 0.55(0.014) & 0.50(0.034) & 38.56(5.11) & \text{438.2(36.9)} & \textbf{352.4(19.8)}  \\ 
                                & 0.9 & 2.73(0.072) & 0.63(0.017) & 0.56(0.041) & 40.32(4.97) & {443.0(41.3)} & \textbf{368.8(22.3)}  \\ \hdashline[0.5pt/5pt]
        \multirow{3}{*}{TADMM} & 0.1 & 4.37(0.269) & 1.64(0.045) & 1.51(0.032) & 23.25(2.41) & 500+(0.00) & 800.3(42.4)  \\ 
                               & 0.5 & 4.25(0.261) & 1.58(0.042) & 1.46(0.030) & 21.73(2.25) & 500+(0.00) & 813.6(50.8)  \\ 
                               & 0.9 & 4.69(0.278) & 1.70(0.048) & 1.55(0.034) & 26.08(2.83) & 500+(0.00) & 798.1(46.3)  \\ \hdashline[0.5pt/5pt]
        \multirow{3}{*}{PPA} & 0.1 &{1.57(0.020)} & {0.42(0.009)} & \textbf{0.30(0.008)} & {12.23(0.81)} & 456(36.3) & 460.2(27.2)  \\ 
                             & 0.5 & {1.62(0.024)} & {0.51(0.010)} & \textbf{0.32(0.009)} & \textbf{11.91(0.75)} & 473(40.5) & 480.4(30.3)  \\ 
                             & 0.9 & \textbf{1.74(0.031)} & {0.63(0.012)} & \textbf{0.35(0.009)} & {12.55(0.89)} & 491(45.8) & 507.2(35.6)   \\ \hline
    \end{tabular}}
\label{tab2}
\end{table*}

\subsection{Parallel environment}\label{sec42}
In the parallel environment, we focus on  denser  coefficients (many coefficients are non-zero) rather than sparse ones. In fact, dense coefficients often appear in numerical experiments related to DS algorithms, as seen in \cite{Lu2012An}, \cite{Wang2012The}, \cite{Prater2015Finding}, \cite{He2015A} and \cite{Mao2021A}. Here, we set $(n, p, |\mathcal{A}|) = (720s, 2560s, 320s)$ for $s = 1, \dots, 10$. Under this configuration, non-zero coefficients constitute one-eighth of the total coefficients. Consequently, we partition $p$ into 80 segments and randomly select 10 segments as non-zero coefficients, each containing $32s$ non-zero entries. These non-zero coefficients are defined as follows:
\begin{equation}\label{gbeta}
\beta^*_j = \left\{ \begin{array}{l}
\xi_j(1 + |a_j|),\ \ \ \ \text{if} \ j \in \mathcal{A}, \\
0,\ \ \ \ \ \ \ \ \ \ \ \ \ \ \ \ \ \text{otherwise},
\end{array} \right.
\end{equation}
where $\xi_j$ is randomly selected from the set $\{+1, -1\}$ and $a_j$ follows a $\mathcal{N}(0, 1)$ distribution. When implementing parallel computing, matrix $\bm A$ will be divided into $K$ blocks. Currently, the only parallel computing algorithm for DS is the three-block ADMM algorithm (TADMM) proposed by \cite{Wen2024Nonconvex}. Therefore, we will only compare the two parallel PPA algorithms presented in this paper with TADMM.

\begin{table*}[!ht]\footnotesize
    \centering
    \caption{Comparison of parallel environment for solving $\ell_1$-DS with $s=5$.}
    \renewcommand\arraystretch{1.5}
    \resizebox{\linewidth}{!}{
    \begin{tabular}{lllllll}
        \hline
        Method & $K$ & \qquad AE & \quad \ FP &   \quad FN  & \quad \ Ite & \quad \ Time(s)  \\ \hline
        \multirow{4}{*}{TADMM} & 1  & 0.412(0.148) &  123.6(14.2) &0.32(0.002) & 500+(0.00) & 2528.6(123.2)  \\ 
                               & 5  & 0.453(0.151) &  131.5(15.6) & 0.36(0.002)&500+(0.00) & 623.4(45.6)  \\ 
                               & 10 & 0.528(0.165) &  137.1(16.0) & 0.47(0.003)&500+(0.00) & 381.9(31.8)  \\ 
                               & 20 & 0.545(0.179) &  141.2(16.3) & 0.51(0.003) & 500+(0.00) & 237.0(22.5)  \\ \hdashline[0.5pt/5pt]
        \multirow{4}{*}{PPPA}  & 1 & \textbf{0.103(0.012)} &  \textbf{85.52(9.2)} & 0.12(0.001) & \textbf{442.3(31.8)} & 994.2(63.1)  \\ 
                               & 5 & \textbf{0.103(0.012)} &  \textbf{85.52(9.2)} & \textbf{0.12(0.001)} & \textbf{442.3(31.8)} & 278.4(23.5)  \\ 
                               & 10 & \textbf{0.103(0.012)} &  \textbf{85.52(9.2)} & \textbf{0.12(0.001)}& \textbf{442.3(31.8)} & \textbf{175.8(16.2)}  \\ 
                               & 20 & \textbf{0.103(0.012)} &  \textbf{85.52(9.2)} & \textbf{0.12(0.001)} & \textbf{442.3(31.8)} & \textbf{132.1(10.4)}  \\ \hdashline[0.5pt/5pt]
        \multirow{4}{*}{IPPPA} & 1 & 0.109(0.013) &   87.28(9.90) & \textbf{0.11(0.001)} &456.6(33.7) & \textbf{902.3(57.9)}  \\ 
                               & 5 & 0.110(0.014) &   89.15(10.3) & 0.13(0.001) & 463.2(34.2) & \textbf{268.6(20.6)}  \\ 
                               & 10 & 0.113(0.014) &  91.34(10.9) & 0.15(0.001) & 466.8(35.1) & 181.5(15.8)  \\ 
                               & 20 & 0.115(0.015) &   93.40(11.5) & 0.18(0.001)& 475.4(36.5) & 147.0(11.5)  \\ \hline
    \end{tabular}}
\label{tab3}
\end{table*}

Unlike the evaluation criteria for sparse coefficients, we use the following five  evaluation criteria in the example of dense coefficients: (1) absolute error (AE),  expressed by \( \lVert \hat{\bm \beta} - \bm \beta^* \rVert_1/p \); (2) false positives (FP), representing the quantity of non-significant features chosen;  (3) false negatives (FN), indicating the number of significant features omitted; (4) number of iterations (NI); and (5) CPU runtime (Time), measured in seconds. To save space, we present results only for \( K = 1, 5, 10, 20 \) and \( \rho = 0.5 \) in Table \ref{tab3} and Table \ref{tab4}.  \textcolor{red}{Note that \( K = 1\) means computing in a nonparallel environment.}  Results for other values of \( K \) will be shown in Figure \ref{Fig1} and Figure \ref{Fig2}. The above numerical experimental results are all for convex $\ell_1$-DS. For numerical experiments on nonconvex DS, such as SCAD-DS and MCP-DS, we have included them in the Appendix \ref{B2}.
\begin{table*}[!ht]\footnotesize
    \centering
    \caption{Comparison of parallel environment for solving $\ell_1$-DS with $s=10$.}
    \renewcommand\arraystretch{1.5}
    \resizebox{\linewidth}{!}{
    \begin{tabular}{lllllll}
        \hline
        Method & $K$ & \qquad AE & \quad \ FP & \quad \ FN &  \quad \ Ite & \quad \ Time(s)  \\ \hline
        \multirow{4}{*}{TADMM} & 1  & 0.833(0.227) & 352.5(30.4) &1.32(0.010) & 500+(0.00) & 5443.7(277.5)  \\ 
                               & 5  & 0.842(0.239) & 361.0(32.7) &1.57(0.016) & 500+(0.00) & 1015.5(152.7)  \\ 
                               & 10 & 0.857(0.248) & 372.2(35.1) & 1.84(0.021) & 500+(0.00) & 692.7(72.9)  \\ 
                               & 20 & 0.866(0.256) & 380.1(38.2) & 2.23(0.027) & 500+(0.00) & 386.2(45.8)  \\ \hdashline[0.5pt/5pt]
        \multirow{4}{*}{PPPA}  & 1  & 0.191(0.024) & \textbf{163.6(15.8)} & 0.51(0.006) & \textbf{472.6(34.9)} & 1624.1(92.5)  \\ 
                               & 5  & 0.191(0.024) & \textbf{163.6(15.8)} & \textbf{0.51(0.006)} & \textbf{472.6(34.9)} & 357.5(35.2)  \\ 
                               & 10 & \textbf{0.191(0.024)} & \textbf{163.6(15.8)} & \textbf{0.51(0.006)} & \textbf{472.6(34.9)} & \textbf{245.9(20.3)}  \\ 
                               & 20 & \textbf{0.191(0.024)} & \textbf{163.6(15.8)} & \textbf{0.51(0.006)} & \textbf{472.6(34.9)} & \textbf{161.2(13.6)}  \\ \hdashline[0.5pt/5pt]
        \multirow{4}{*}{IPPPA} & 1  & \textbf{0.185(0.022)} & 172.3(16.4) & \textbf{0.47(0.06)} & 487.8(40.1) & \textbf{1476.5(98.9)}  \\ 
                               & 5  & \textbf{0.190(0.023)} & 179.2(17.2) & 0.53(0.07)& 498.5(43.7) & \textbf{301.4(29.7)}  \\ 
                               & 10 & 0.196(0.024) & 184.6(18.5) & 0.61(0.08) & 500+(50.2) & 257.8(22.1)  \\ 
                               & 20 & 0.199(0.026) & 190.2(20.7) & 0.64(0.08) & 500+(55.3) & 212.52(19.2)  \\ \hline
    \end{tabular}}
\label{tab4}
\end{table*}

An interesting observation in two figures is that in Figure  \ref{Fig1}, when \( K \geq 20 \), increasing the number of partitions \( K \) of the gram matrix does not accelerate the computation speed for three parallel algorithms. In Figure \ref{Fig2}, this phenomenon appears when \( K \geq 30 \).  The observed phenomenon is likely a consequence of the combined effects of increased synchronization and communication overhead, potential load imbalance, memory bandwidth and cache constraints, non-linear scaling of algorithmic complexity, and hardware limitations in handling a large number of parallel tasks. Nevertheless, moderately and effectively increasing the number of parallelism is an effective way to accelerate the computation speed of DSs.
In Tables \ref{tab3} and \ref{tab4}, with the variation of \(K\), all metrics in PPPA, excluding the Time metric, \textcolor{red}{remain constant and are identical to those when \(K = 1\) (nonparallel environment).    The reduction in computation time can be attributed to two factors. Firstly, matrix partitioning shortens the time required for matrix-matrix or matrix-vector multiplications. Secondly, the iterative variables corresponding to each sub-matrix are processed in parallel. } This is also the partition insensitivity described in Theorem \ref{TH1}. This insensitivity is also reflected in Figures \ref{Fig1} and \ref{Fig2}, where the AE of PPA is shown as a horizontal line. The two PPA algorithms outperform TADMM in terms of estimation accuracy, variable selection, and computational efficiency. As we discussed earlier, our algorithms have fewer iterative variables, which not only enhances the accuracy of the solutions but also improves computational speed. 

\begin{figure}[H]
\centering
\setlength{\abovecaptionskip}{0.cm}  
\subfigure{
\label{Fig1.a}
\includegraphics[width=0.48\textwidth]{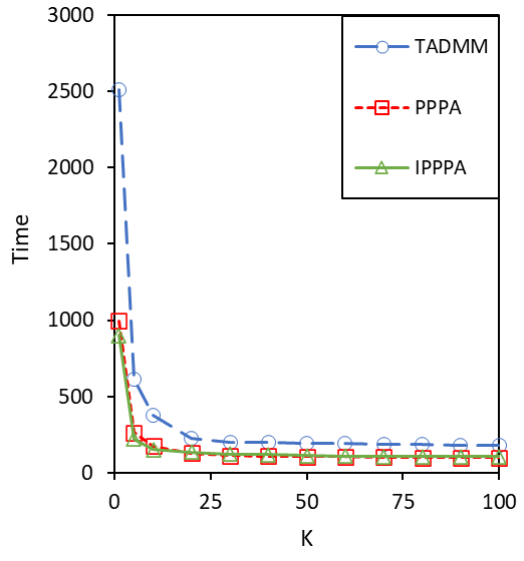}}
\subfigure{
\label{Fig1.b}
\includegraphics[width=0.48\textwidth]{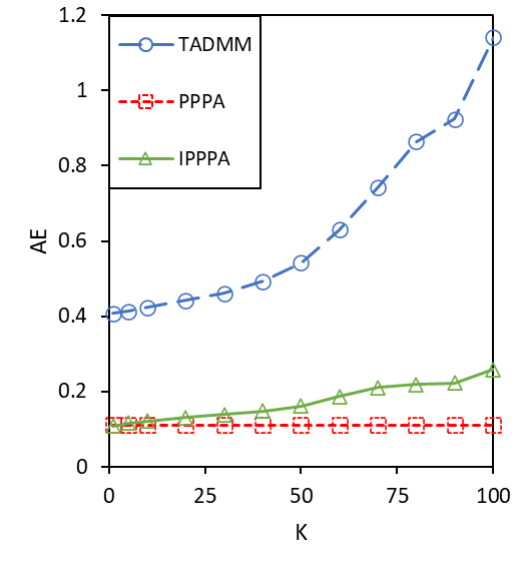}}
\caption{A schematic diagram illustrating the variation of AE and Time with respect to \( K \), where \( s = 5 \).}
\label{Fig1}
\end{figure}

\begin{figure}[H]
\centering
\setlength{\abovecaptionskip}{0.cm}  
\subfigure{
\label{Fig2.a}
\includegraphics[width=0.48\textwidth]{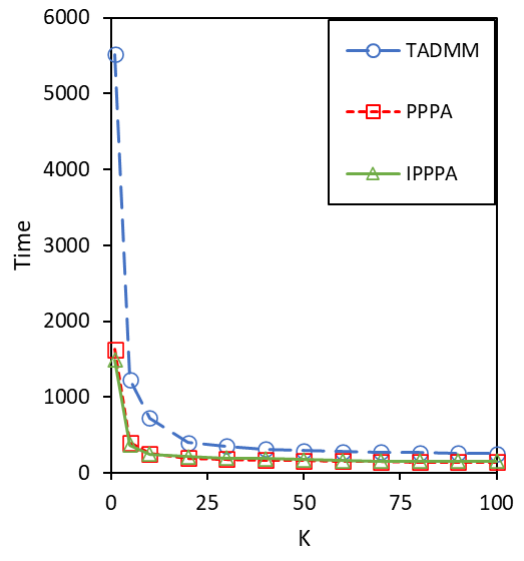}}
\subfigure{
\label{Fig2.b}
\includegraphics[width=0.48\textwidth]{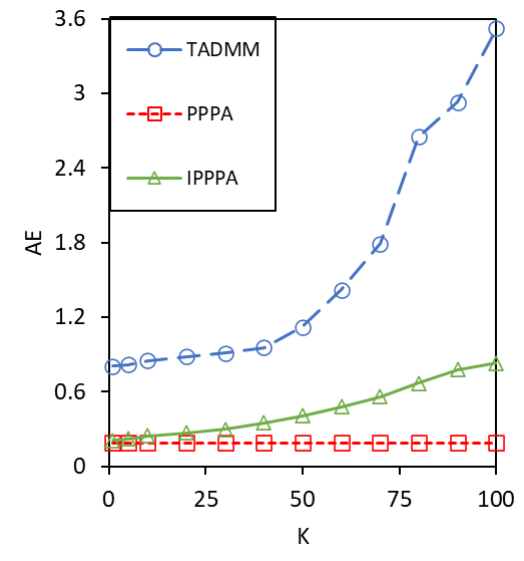}}
\caption{A schematic diagram illustrating the variation of AE and Time with respect to \( K \), where \( s = 10 \).}
\label{Fig2}
\end{figure}

The numerical performance of the improved parallel PPA (IPPPA) compared to the parallel PPA (PPPA) presented in Tables \ref{tab3} and \ref{tab4} does not appear to have improved. In many cases, IPPPA  may require even more iterations and longer computation times. The reason IPPPA outperforms PPPA in computational effectiveness is that IPPPA approximates \( \text{eigen}(\mu \bm A^\top \bm A) \) by calculating the eigenvalues of submatrices \( \text{eigen}(\mu \bm A_i^\top \bm A_i) \) instead of directly computing \( \text{eigen}(\mu \bm A^\top \bm A) \). This approach leads to additional iterations, and indeed, as \( K \) increases, the computational performance of IPPPA deteriorates and its efficiency decreases. The real advantage of IPPPA over PPPA comes into play when the storage requirements for the matrix are too large to allow for \( \text{eigen}(\mu \bm A^\top \bm A) \)  calculations; in such cases, PPPA becomes infeasible, while IPPPA can still be implemented.

\section{Numerical simulations using real dataset}\label{sec5}
In this section, we utilize two real-world datasets to construct three different DS models: $\ell_1$-DS, SCAD-DS, and MCP-DS. These models are trained using the parallel computation algorithms TADMM, proposed by \cite{Wen2024Nonconvex}, as well as two parallel PPA algorithms introduced in this paper.
\subsection{Leukemia dataset}
In this experiment, the convex and nonconvex Dantzig selectors generated by various algorithms are applied to a set of biomarker data to determine whether a patient may be diagnosed with a specific type of cancer. Leukemia  is a type of cancer that affects the blood or bone marrow and is characterized by an abnormal increase in white blood cells. The leukemia dataset first introduced in \cite{Golub1999Molecular}  consists of 38 training samples and 34 test samples. Among the 38 training samples, 27 are classified as acute lymphocytic leukemia (ALL), while the remaining 11 are classified as acute myelogenous leukemia (AML). Each sample in this dataset includes measurements of $7,129$ genes.  This data has been studied in the literature through various regularization methods and Dantzig selectors, such as in \cite{Lu2012An}, \cite{Wang2012The} \cite{Prater2015Finding} and \cite{Wu2024Multi}. 
\begin{table}[!ht]\footnotesize
    \centering
    \caption{Comparison of three DSs for leukaemia data.}
    \renewcommand\arraystretch{1.5}
    \begin{tabular}{ccccccc}
    \hline
        \makebox[0.12\textwidth][c]{Method} & \makebox[0.13\textwidth][c]{Algorithm} & \makebox[0.1\textwidth][c]{TRE} & \makebox[0.1\textwidth][c]{TEE} &  \makebox[0.1\textwidth][c]{$| \hat{\mathcal A}|$}  & \makebox[0.1\textwidth][c]{Ite} & \makebox[0.1\textwidth][c]{Time(s)} \\ \hline
        \multirow{3}{*}{$\ell_1$-DS} & TADMM & \textbf{0/38} & 2/34 & 206 & 182 & 1.96 \\ 
                               & PPPA & \textbf{0/38} & \textbf{0/34} & \textbf{63} & \textbf{120} & \textbf{1.21} \\ 
                               & IPPPA & \textbf{0/38} & \textbf{0/34} & 67 & 123 & 1.26 \\ \hdashline[0.5pt/5pt]
        \multirow{3}{*}{SCAD-DS} & TADMM & \textbf{0/38} & \textbf{0/34} & 105 & 193 & 2.07 \\ 
                                 & PPPA & \textbf{0/38} & \textbf{0/34} & 44 & 126 & 1.28 \\ 
                                 & IPPPA & \textbf{0/38} & \textbf{0/34} & \textbf{42} & \textbf{125} & \textbf{1.24} \\ \hdashline[0.5pt/5pt]
        \multirow{3}{*}{MCP-DS} & TADMM & \textbf{0/38} & 1/34 & 109 & 195 & 2.11 \\ 
                                & PPPA & \textbf{0/38} & \textbf{0/34} & \textbf{49} & \textbf{127} & \textbf{1.30} \\ 
                                & IPPPA & \textbf{0/38} & \textbf{0/34} & 50 & 130 & 1.33 \\ 
    \hline
    \end{tabular}
    \label{tab5}
\end{table}

Let \( \bm{X}_{\text{train}} \in \mathbb{R}^{38 \times 7129} \) represent the leukemia dataset from the training set, where each row corresponds to the $7,129$ gene measurements of a single patient, and each column has been normalized to have unit \(\ell_2\) norm. Let \( \bm{y}_{\text{train}} \in \mathbb{R}^{38} \) denote the column vector indicating the diagnosis of each patient within the training set, defined as:
\[
\bm{y}_{\text{train}}(j) = 
\begin{cases} 
1, & \text{if patient } j \text{ in the training set is diagnosed with AML}, \\
0, & \text{if patient } j \text{ in the training set is diagnosed with ALL}.
\end{cases}
\]
Similarly, define \( \bm{X}_{\text{test}} \in \mathbb{R}^{34 \times 7129} \) and \( \bm{y}_{\text{test}} \in \mathbb{R}^{34} \) using the data from the testing set.

During the testing phase, the trained parameter \(\hat{\bm \beta}\) is utilized to predict the diagnoses of patients in the testing set. The predictive indicator vector \(\hat{\bm{y}}_{\text{test}} \in \mathbb{R}^{34}\) is computed from \( \bm{y} = \bm{X}_{\text{test}} \hat{\bm \beta} \) by applying thresholding and clustering to identify values near the threshold boundary.
Define \(\hat{\bm{y}}_{\text{test}}(j)\) as follows:
\[
\hat{\bm{y}}_{\text{test}}(j) = 
\begin{cases} 
1, & \text{if } y(j) \geq 0.5, \\
0, & \text{if } y(j) < 0.49.
\end{cases}
\]
This allows us to categorize the predictions based on whether they fall above or below the specified thresholds.

In a parallel computing environment, we randomly partition the $7,129$ features into 5 groups, denoted as \( K = 5 \), with 4 groups containing $1,425$ features each and one group containing $1,429$ features. In  Table \ref{tab5}, we present the results of the numerical experiments, where TRE represents the training error, TEE denotes the testing error, and $|\hat{\mathcal A}|$ indicates the number of estimated coefficients that are non-zero.  Overall, the numerical results indicate that, in a parallel computing environment, the two parallel PPA algorithms yield lower prediction errors and require fewer iterations compared to TADMM when solving the DS models constructed using leukemia dataset. This substantial reduction in iterations significantly saves computational time. 

\textcolor{red}{Specifically, in terms of TRE metric, all three parallel algorithms can accurately identify symptoms of ALL and AML. However, in terms of TEE metric, the parallel PPA algorithm proposed in this paper will perform better than the TADMM algorithm. In details, on the prediction set,  $\ell_1$-DS model trained by the PPA algorithms accurately identifies 20 cases of ALL and 14 cases of AML. In contrast, $\ell_1$-DS model trained by the TADMM algorithm misidentifies one case of AML as ALL and one case of ALL as AML. That is to say, TEE is 2 out of 34. The SCAD-DS model trained by the three algorithms can accurately recognize all ALL and AML.  For the MCP-DS model trained by the three algorithms, the model trained by PPAs can accurately recognize all ALL and AML, but TADMM will incorrectly recognize one AML, which has a TEE of 1/34. More importantly, the two PPA algorithms select a smaller number of non-zero coefficients. In other words, the value of $|\hat{\mathcal A}|$ for the PPA algorithms is smaller than that for the TADMM algorithm. This characteristic holds a certain guiding significance in practical applications. It implies that in the future, one only needs to observe the features corresponding to the non-zero coefficients in each sample. Consequently, this approach can help save the cost of data collection. }

\subsection{Supermarket dataset}
In this subsection, we assess the performance of the proposed algorithms using a supermarket dataset previously analyzed by 
\cite{Wang2009Forward}. This data was also evaluated by \cite{Wen2024Nonconvex} for their proposed regularization method and algorithm. This dataset includes the daily number of customers and the daily sales volumes of \( p = 6,398 \) products from a Chinese supermarket over a period of \( n = 464 \) days. \textcolor{red}{The supermarket manager is keen to determine which product's sales volume exhibits the strongest correlation with the number of customers, after accounting for the influence of other products. Then, DS models prove to be valuable. } Hence, our goal is to predict the daily number of customers (the response variable of interest) using the daily sales volumes as predictors. In line with \cite{Wang2009Forward}, both the response variable and predictors have been standardized to have a mean of zero and a variance of one for confidentiality reasons.

\begin{table}[!ht]\footnotesize
    \centering
    \caption{Comparison of three DSs for supermarket data.}
    \renewcommand\arraystretch{1.5}
    \begin{tabular}{ccccccc}
    \hline
        \makebox[0.12\textwidth][c]{Method} & \makebox[0.13\textwidth][c]{Algorithm} & \makebox[0.1\textwidth][c]{TRE} & \makebox[0.1\textwidth][c]{TEE} &  \makebox[0.1\textwidth][c]{$|\hat{\mathcal A}|$}  & \makebox[0.1\textwidth][c]{Ite} & \makebox[0.1\textwidth][c]{Time(s)}  \\ \hline
        \multirow{3}{*}{$\ell_1$-DS} & TADMM & 0.232 & 0.329 & 172 & 223 & 2.53  \\ 
                               & PPPA & 0.227 & \textbf{0.275} & \textbf{96} & \textbf{141} & \textbf{1.59}  \\ 
                               & IPPPA & \textbf{0.225} & 0.298 & 103 & 152 & 1.62  \\ \hdashline[0.5pt/5pt]
        \multirow{3}{*}{SCAD-DS} & TADMM & 0.226 & 0.349 & 41 & 271 & 2.98  \\ 
                                 & PPPA & 0.218 & \textbf{0.312} & 39 & \textbf{149} & \textbf{1.70}  \\ 
                                 & IPPPA & \textbf{0.213} & 0.320 & \textbf{36} & 163 & 1.81  \\ \hdashline[0.5pt/5pt]
        \multirow{3}{*}{MCP-DS} & TADMM & 0.231 & 0.332 & 50 & 268 & 2.85  \\ 
                                & PPPA & \textbf{0.219} & \textbf{0.316} & \textbf{43} & \textbf{153} & \textbf{1.72}  \\ 
                                & IPPPA & 0.227 & 0.325 & 41 & 167 & 1.84  \\ \hline
    \end{tabular}
    \label{tab6}
\end{table}
Following \cite{Wen2024Nonconvex}, we divide the dataset into a training set comprising observations from the first 300 days and a testing set consisting of the remaining days. Firstly, using the training dataset, we construct three DS models: $\ell_1$-DS, SCAD-DS, and MCP-DS. Next, we employ the TADMM proposed by \cite{Wen2024Nonconvex} and the two parallel PPA algorithms presented in this paper to solve these models. The obtained estimated coefficients are then used for prediction.

In a parallel computing environment, we randomly partition the  \( p = 6,398 \) features into 5 groups, denoted as \( K = 5 \), with 4 groups containing $1,279$ features each and one group containing $1,282$ features.  We present the results of the numerical experiments in the Table \ref{tab6}, where TRE represents the training error, TEE denotes the testing error, and  $|\hat{\mathcal A}|$  indicates the number of estimated coefficients that are non-zero. The numerical results in Table \ref{tab6} indicate that the two parallel PPA algorithms have better training and testing errors compared to TADMM when solving the DS models constructed with supermarket dataset, and also have improved computational efficiency.    \textcolor{red}{To be specific,  two PPA algorithms offer more accurate predictions of the daily number of customers compared to the TADMM algorithm. This means that supermarket managers can utilize our algorithms to obtain more precise forecasts, enabling them to better arrange the quantity of goods and the number of service staff.  In addition, the coefficients estimated by our algorithm have a higher degree of sparsity. When the variety of supermarket products remains unchanged, this can reduce the cost of data collection in subsequent data-gathering processes. }

\section{Conclusion and further research}\label{sec6}
In this paper, we have successfully developed a novel variable splitting parallel algorithm for addressing both convex and nonconvex Dantzig selectors, leveraging the proximal point algorithm. Our approach uniquely minimizes the number of iteration variables, thereby significantly enhancing computational efficiency and accelerating convergence. Notably, our algorithm exhibits a partition-insensitive property, ensuring consistent performance regardless of data partitioning. Theoretical analysis has confirmed the linear convergence of our algorithm, and empirical results demonstrate its competitive performance across various computational environments. Additionally, to handle extreme situations where the matrix dimension is excessively large, and the maximum eigenvalue of the entire gram matrix cannot be computed, we have developed an enhanced version of this parallel PPA algorithm.
The R package for our algorithm, available at \url{https://github.com/xfwu1016/PPADS}, invites contributions and additional applications from the research community.

While our work has shown promising results, future research could explore several avenues for further enhancement. Firstly, investigating adaptive partitioning strategies to optimize performance across diverse datasets could yield even greater efficiency gains. Secondly, extending the algorithm to handle streaming data scenarios, where data arrives sequentially, would broaden its applicability. Lastly, integrating our method with advanced machine learning models and leveraging GPU acceleration for even faster processing could further solidify its utility in large-scale data applications.

\section*{Acknowledgements}
We would like to extend our sincere gratitude to Professor Bingsheng He from Nanjing University for sharing his work on the proximal point algorithm. It is his pioneering research that inspired this paper.  The research of Zhang was supported by the National Natural Science Foundation of China [Grant Numbers 12271066, 12171405, 11871121], and the research of Wu, Tang and Liang  was supported by the Scientific and Technological Research Program of Chongqing Municipal Education Commission [Grant Numbers KJQN202302003].

\begin{appendices}
\section{\textcolor{red}{Proofs of  Theorems \ref{TH1} and  \ref{TH5} }} \label{C}

\subsection{Proof of  Theorem \ref{TH1}}\label{C1}
\textcolor{red}{Recall that we denote $\left\{ \hat{\bm \beta}^t, \hat{\bm z}^t, \hat{\bm u}^t\right\}$ as the results of the $t$-th iteration of Algorithm \ref{alg1}, and $\left\{ \tilde{\bm \beta}^t, \tilde{\bm z}^t, \tilde{\bm u}^t \right\}$ as the results of the $i$-th iteration of Algorithm \ref{alg2}. We make the assumption that 
\begin{align}\label{ass}
\left\{ \hat{\bm \beta}^t, \hat{\bm z}^t, \hat{\bm u}^t\right\} = \left\{ \tilde{\bm \beta}^t, \tilde{\bm z}^t, \tilde{\bm u}^t \right\}, 
\end{align}
and subsequently, we will conduct an analysis of the iteration scenario at step $t + 1$. }

\textcolor{red}{Firstly, we will discuss the update of the $\bm \beta$-subproblem. For Algorithm \ref{alg1}, we have 
\begin{align}\label{prc1}
\hat{\bm \beta}^{t+1} \leftarrow   \mathop {\arg \min }\limits_{\bm \beta} \left\{\| \bm \beta \|_1 + \frac{\eta}{2} \|\bm \beta - \hat{\bm \beta}^t - \frac{\bm A^\top \hat{\bm u}^t}{\eta}   \|_{2}^2 \right \};
\end{align}
and  for   Algorithm \ref{alg2}, we have
\begin{align}\label{prc2}
\tilde{\bm \beta}_{i\cdot} ^{t+1} \leftarrow  \mathop {\arg \min }\limits_{\bm \beta_{i\cdot}} \left\{\| \bm \beta_{i\cdot} \|_1 + \frac{\eta}{2} \|\bm \beta_{i\cdot} - \tilde{\bm \beta}_{i\cdot}^t - \frac{\bm A_i^\top \tilde{\bm u}^t}{\eta}   \|_{2}^2 \right \},  i = 1,2,\dots,K.
\end{align}
If  $\hat{\bm \beta}^{t+1}$ is also decomposed from the column,  that is,  $$\hat{\bm \beta}^{t+1} = ((\hat{\bm \beta}_{1\cdot}^{t+1})^\top, (\hat{\bm \beta}_{2\cdot}^{t+1})^\top, \dots, (\hat{\bm \beta}_{K\cdot}^{t+1})^\top)^\top.$$
Since $\bm A = [\bm A_1, \bm A_2,\dots,\bm A_K ]$,  it can be concluded that \begin{align}\label{prc3}
\hat{\bm \beta}_{i\cdot} ^{t+1} \leftarrow  \mathop {\arg \min }\limits_{\bm \beta_{i\cdot}} \left\{\| \bm \beta_{i\cdot} \|_1 + \frac{\eta}{2} \|\bm \beta_{i\cdot} - \hat{\bm \beta}_{i\cdot}^t - \frac{\bm A_i^\top \hat{\bm u}^t}{\eta}   \|_{2}^2 \right \},  i = 1,2,\dots,K.
\end{align}
According to the assumption in \eqref{ass}, there is  $\hat{\bm \beta}_{i\cdot}^t =  \tilde{\bm \beta}_{i\cdot}^t $ and  $\hat{\bm u}^t = \tilde{\bm u}^t$. It follows from \eqref{prc2} and  \eqref{prc3} that
\begin{align}\label{prbeta}
 \hat{\bm \beta}_{i\cdot} ^{t+1} = \tilde{\bm \beta}_{i\cdot} ^{t+1} ,  i = 1,2,\dots,K.
\end{align}}

\textcolor{red}{Next, we will discuss the update of the $\bm z$-subproblem.  For Algorithm \ref{alg1}, we have   \begin{align}\label{prz}
\hat{\bm z}^{t+1} \leftarrow  \mathop {\arg \min }\limits_{\bm z} \left\{ \delta_{\mathcal{Z}_0(\bm z)} + \frac{\mu}{2} \|\bm z - \hat{\bm z}^t + \frac{ \hat{\bm u}^t}{\mu}   \|_{2}^2 \right \};
\end{align}
and   for   Algorithm \ref{alg2}, we get
\begin{align}\label{prz2}
\tilde{\bm z}^{t+1} \leftarrow  \mathop {\arg \min }\limits_{\bm z} \left\{ \delta_{\mathcal{Z}_0(\bm z)} + \frac{\mu}{2} \|\bm z - \tilde{\bm z}^t + \frac{ \tilde{\bm u}^t}{\mu}   \|_{2}^2 \right \};
\end{align}
Again,  according to the assumption in \eqref{ass},   there is  $\hat{\bm z}^t = \tilde{\bm z}^t$ and  $\hat{\bm u}^t = \tilde{\bm u}^t$.  It follows from \eqref{prz} and  \eqref{prz2} that
\begin{align}\label{prz3}
 \hat{\bm z}^{t+1} = \tilde{\bm z}^{t+1}.
\end{align}
Finally, we will discuss the update of the $\bm u$-subproblem. Since $\sum_{i=1}^{K} \bm A_i \bm \beta_{i\cdot}  = \bm A \bm \beta $,   It follows from \eqref{prbeta} and  \eqref{prz3} that 
\begin{align}\label{pru}
 \hat{\bm u}^{t+1} = \tilde{\bm u}^{t+1}.
\end{align}}

\textcolor{red}{Therefore,   we can draw the conclusion that when the iteration sequences of the two algorithms at step $t$ are the same, the iteration sequences at step $t + 1$  will also be the same.  Based on this, assuming that \(\{ \hat{\bm \beta}^0, \hat{\bm z}^0, \hat{\bm u}^0 \} = \{ \tilde{\bm \beta}^0, \tilde{\bm z}^0 , \tilde{\bm u}^0 \}\), it follows that \(\{ \hat{\bm \beta}^1, \hat{\bm z}^1, \hat{\bm u}^1 \} = \{ \tilde{\bm \beta}^1, \tilde{\bm z}^1 , \tilde{\bm u}^1 \}\). By the same token, we can deduce that 
\begin{align}
\left\{ \hat{\bm \beta}^t, \hat{\bm z}^t, \hat{\bm u}^t \right\} = \left\{ \tilde{\bm \beta}^t, \tilde{\bm z}^t , \tilde{\bm u}^t \right\}, \quad \text{for all } t.
\end{align}
Up to this point, the proof of the theorem has been successfully completed. }%

\subsection{Proof of  Theorem \ref{TH5}}\label{C2}

\textcolor{red}{This proof is divided into two parts, first proving the time complexity of the algorithm, and then proving the space complexity of the algorithm. }

\textcolor{red}{(1). Reviewing the iterative processes of Algorithm \ref{alg2} and Algorithm \ref{alg3}, we observe that the updates of the  \(\bm \beta\),  \(\bm z\) and \(\bm u\) subproblems are nearly identical. The main computational burden of \(\bm \beta\) updates is the calculation of $\bm A^\top \bm u$.  Due to its parallel computation across $K$ nodes, each node only needs to calculate $\bm A_{i}^\top \bm u$.   Its computational complexity is $\mathcal{O}(\frac{p^2}{K})$. The update of the $\bm z$ subproblem only involves the addition, subtraction, and comparison of some $p$-dimensional vectors. Then, its computational complexity is $\mathcal{O}(p)$. The biggest computational burden of updating \(\bm u\)  lies in computing $\bm A \bm \beta = \sum_{i=1}^{K} \bm A_i \bm \beta_{i\cdot}$, which will be computed in parallel across $K$ nodes. Its computational complexity is $\mathcal{O}(\frac{p^2}{K})$. Clearly,  $\mathcal{O}(\frac{p^2}{K}) + \mathcal{O}(p) + \mathcal{O}(\frac{p^2}{K}) =  \mathcal{O}(\frac{p^2}{K})$ because  the number of column partitions cannot exceed the number of columns, i.e. $p \ge K$.}

\textcolor{red}{Next,  we need to discuss the time complexity of the pre-computation for Algorithm \ref{alg2} and Algorithm \ref{alg3} separately.   Algorithm \ref{alg2}  requires pre-computation of $\bm X^\top \bm y$,  $\bm A_i =  \bm X^\top \bm X_i$,  and calculate $\eta$ using the power method. The time complexity required for the first and second items is  $\mathcal O(np)$ and  $\mathcal O(\frac{np^2}{K}) $, respectively. As discussed in \cite{Liang2024Linearized}, the computational complexity required for the power method is quadratic with the number of columns in matrix $\bm A$, i.e. $\mathcal O({p^2}) $. Thus,   the time complexity of the pre-computation for Algorithm \ref{alg2} is  $\mathcal O(\frac{np^2}{K}) + \mathcal O({p^2})  $. The main difference between Algorithm  \ref{alg3} pre-computation and Algorithm \ref{alg2} is that one calculates $\eta$ and the other calculates $\eta_i$. Similarly, the complexity of calculating $\eta_i$ is $\mathcal O(\frac{p^2}{K^2})$.  Thus,   the time complexity of the pre-computation for Algorithm \ref{alg3} is  $\mathcal O(\frac{np^2}{K})  $. 
By summarizing the above discussion, we can arrive at the conclusion in (\ref{comp1}).}

\textcolor{red}{(2). Next, we will discuss the space complexity of the Algorithm \ref{alg2} and Algorithm \ref{alg3}. The space complexity of algorithms mainly includes three parts: the required storage for inputs and outputs, and the additional storage required for computation. The maximum storage space required for the input and output of Algorithm  \ref{alg2} is matrix $\bm X$, which is $\mathcal O(np)$. The largest additional storage during the precomputation process is the storage of matrix $\bm A$, which is $\mathcal O(p^2)$. The additional storage required for updating the three subproblems is relatively small, which is $\mathcal O(p)$. Therefore, the required storage space for Algorithm \ref{alg2}  is  $\mathcal O(np) + \mathcal O(p^2)$.}

\textcolor{red}{Similarly, the maximum storage space required for the input and output of Algorithm  \ref{alg3} is matrix $\bm X$, which is $\mathcal O(np)$. he largest additional storage during the precomputation process is the storage of matrix $\bm A_i$, which is $\mathcal O(\frac{p^2}{K})$.  The updates of the three sub problems in Algorithm \ref{alg3} are similar to those in Algorithm \ref{alg2}, requiring only $\mathcal O(p)$ of additional storage space. Therefore, the required storage space for Algorithm \ref{alg3}  is  $\mathcal O(np) +\mathcal O(\frac{p^2}{K})$.}

\section{Proofs of  Theorems \ref{TH2} and \ref{TH3}}\label{A}

\subsection{preliminary}

\subsubsection{Lemma}
\begin{lem}\label{lem1}
(Lemma 2.1 in \cite{He2022A}). Let \( \mathbb{X} \subset \mathbb{R}^l \) be a closed convex set, and let \( \theta : \mathbb{R}^l \to \mathbb{R} \) and \( f : \mathbb{R}^l \to \mathbb{R} \) be convex functions. Suppose \( f \) is differentiable on an open set containing \( \mathbb{X} \), and the minimization problem
\[
\min \{ \theta(x) + f(x) \mid x \in \mathbb{Z} \}
\]
has a nonempty solution set. Then, \( x^* \in \arg \min \{ \theta(x) + f(x) \mid x \in \mathbb{X} \} \) if and only if
\[
x^* \in \mathbb{X} \quad \text{and} \quad \theta(x) - \theta(x^*) + (x - x^*)^\top \nabla f(x^*) \geq 0, \quad \forall x \in \mathbb{X}.
\]
\end{lem}

\begin{lem}\label{lem2}
Assume that  \( \bm{H} \in \mathbb{R}^{n \times n} \) is a positive definite matrix, and \( \bm{a}, \bm{b} \) are two arbitrary \( n \)-dimensional vectors. If $\bm b^\top \bm H (\bm a - \bm b )\ge 0$, then we can have
\begin{align*}
    \| \bm b \|_{\bm H}^2 \le   \| \bm a \|_{\bm H}^2  -  \| \bm a- \bm b \|_{\bm H}^2. 
\end{align*}
   \end{lem}

The conclusion of Lemma \ref{lem2} is very easy to verify. Though simple, this conclusion is widely used to prove PPA and ADMM convergence, as seen in \cite{Cai2013A}, \cite{Gu2014Customized}, \cite{He2018A}, and \cite{He2022A}.
\subsubsection{Four matrices}
To simplify the presentation of analysis, we  define the following four matrices in (\ref{m1}) and (\ref{m2}).
\begin{align}\label{m1}
\bm M = \begin{pmatrix}
\mu \bm A^\top \bm A  & \bm 0  & \bm A^\top \\
\bm 0 & \mu \bm I_p   &  -\bm I_p \\
\bm A &  -\bm I_p & \frac{2}{\mu} \bm I_p
\end{pmatrix},
\
\bm H = \begin{pmatrix}
\mu \bm A^\top \bm A + \bm S  & \bm 0  & \bm A^\top \\
\bm 0 & \mu \bm I_p   &  -\bm I_p \\
\bm A &  -\bm I_p & \frac{2}{\mu} \bm I_p
\end{pmatrix},
\end{align}
where $\bm S= \eta \bm I_p - \mu \bm A^\top \bm A$, and  $\eta >\text{eigen}(\mu  \bm A ^\top  \bm A)$.
\begin{align}\label{m2}
\bm M_K = \begin{pmatrix}
\bm \Lambda_K & \bm G^\top \\
\bm G &  \frac{K+1}{\mu} \bm I_p
\end{pmatrix},
\
\bm H_K = \begin{pmatrix}
\tilde{\bm \Lambda}_K & \bm G^\top \\
\bm G &  \frac{K+1}{\mu} \bm I_p
\end{pmatrix},
\end{align}
 where $\bm S_i^{'} = \eta_i \bm I_p - \mu \bm A_i^\top \bm A_i$,   $\eta_i > \text{eigen}(\mu  \bm A_i ^\top  \bm A_i)$ and $i \in \{1,\dots,K \}$.  Here, $$\bm G = ( \bm A_1,\dots, \bm A_K,-\bm I_p),$$
\[
\bm \Lambda_K = \text{diag}(\mu \bm A_1^\top \bm A_1, \dots,\mu \bm A_K^\top \bm A_K, \mu \bm I_p),
\] and \[
\tilde{\bm \Lambda}_K  = \text{diag}(\mu \bm A_1^\top \bm A_1 +  \bm S_1^{'} , \dots,\mu \bm A_K^\top \bm A_K +  \bm S_K^{'}, \mu \bm I_p) = \text{diag}( \eta_1 \bm I_p, \dots,\eta_K \bm I_p, \mu \bm I_p). 
\]

The two matrices in (\ref{m1}) are utilized for the proof of Theorem \ref{TH2}, while the two matrices in (\ref{m2}) are employed for the proof of Theorem \ref{TH3}. Among them, \(\bm M\) and \(\bm M_K\) are positive semidefinite matrices, and \(\bm H\) and \(\bm H_K\) are positive definite matrices. We employ the following proposition to elucidate this concept.
\begin{prop}
$\bm M \succeq  \bm 0$,  $\bm M_K \succeq  \bm 0$, $\bm H \succ  \bm 0$ \text{and} $\bm H_K \succ  \bm 0$.
\end{prop}
\begin{proof}
\( \bm M \) and \(\bm M_K \) can be represented as $$
  \begin{pmatrix}       
\sqrt{\mu} \bm A^\top \\
\bm 0 \\
\frac{1}{\sqrt \mu} \bm I_p
\end{pmatrix} \begin{pmatrix}       
\sqrt{\mu} \bm A &\bm 0 &\frac{1}{\sqrt \mu} \bm I_p
\end{pmatrix} +  \begin{pmatrix}     
\bm 0 \\  
-\sqrt{\mu} \bm I_p \\
\frac{1}{\sqrt \mu} \bm I_p
\end{pmatrix} \begin{pmatrix}       
\bm 0 &-\sqrt{\mu} \bm I_p &\frac{1}{\sqrt \mu} \bm I_p
\end{pmatrix}, $$ and $$  \sum_{i=1}^K \left[ \begin{pmatrix}       
\vdots\\
\sqrt{\mu} \bm A_i^\top \\
\vdots \\
\frac{1}{\sqrt{\mu}} \bm I_p
\end{pmatrix} \begin{pmatrix}       
\cdots & \sqrt{\mu} \bm A_i&\cdots&\frac{1}{\sqrt{\mu}} \bm I_p
\end{pmatrix} \right] +  \begin{pmatrix}     
\vdots \\  
-\sqrt{\mu} \bm I_p \\
\frac{1}{\sqrt \mu} \bm I_p
\end{pmatrix} \begin{pmatrix}       
\cdots &-\sqrt{\mu} \bm I_p &\frac{1}{\sqrt \mu} \bm I_p
\end{pmatrix}, $$ respectively.
 For any non-zero vector \( \bm{v} \), we have \( \bm{v}^T \bm M \bm{v} \ge 0 \)  and \( \bm{v}^T \bm M_K \bm{v} \ge 0 \), thus $\bm M \succeq  \bm 0$ and  $\bm M_K \succeq  \bm 0$.
It is evident that \( \bm M \) and $\bm M_K$ are at least positive semidefinite matrices; if \( \bm A \) and \( \bm A_i \) have full column rank, then \( \bm M \) and $\bm M_K$ are positive definite matrices.

Similarly, \( \bm H \) and \(\bm H_K \) can be represented as $$  \bm H =  \bm M +   \text{diag}(\bm S, \bm 0, \bm 0), $$ and
$$  \bm H_K =  \bm M_K +   \text{diag}(\bm S_1^{'},  \cdots, \bm S_K^{'}, \bm 0, \bm 0),$$
respectively.
 For any non-zero vector \( \bm{v} \), we have \( \bm{v}^T \bm H \bm{v} > 0 \)  and \( \bm{v}^T \bm H_K \bm{v} > 0 \), thus $\bm H \succ  \bm 0$ and  $\bm H_K \succ  \bm 0$.
\end{proof}

\subsubsection{Variational inequality characterization}
The constrained optimization form of DS  is defined as,
\begin{align}\label{dsp}
\begin{aligned}
& \min_{\bm \beta, \bm z} \quad   \sum_{i=1}^{K}\|\bm \beta_{i\cdot}\|_{1} +  \delta_{\mathcal{Z}_0(\bm z)}  \\
& \text{s.t.} \quad  \sum_{i=1}^{K} \bm A_i \bm \beta_{i\cdot} -  \bm z = \bm X^\top \bm y,
\end{aligned}
\end{align}
And the Lagrange multiplier form of (\ref{dsp}) is   
\begin{align}\label{dsp2}
L(\bm \beta, \bm z;  \bm u) = 
 \sum_{i=1}^{K}\|\bm \beta_{i\cdot}\|_{1}  + \delta_{\mathcal{Z}_0(\bm z)}  - \bm u^\top (  \sum_{i=1}^{K} \bm A_i \bm \beta_{i\cdot}  -  \bm z   -  \bm X^\top \bm y ).
  \end{align}
Since \(  \sum_{i=1}^{K}\|\bm \beta_{i\cdot}\|_{1}  =  \|\bm \beta\|_1\) and  $\sum_{i=1}^{K} \bm A_i \bm \beta_{i\cdot}  = \bm A \bm \beta $, the parallel and nonparallel constraint optimization forms of DS are the same.

According to the variational inequality characterization in section 2 of \cite{He2012On}, we know that the solution of the constrained optimization function above is the saddle point of the  following Lagrangian function in (\ref{dsp2}). 
However, $\sum_{i=1}^{K}\|\bm \beta_{i\cdot}\|_{1}$  and  $\delta_{\mathcal{Z}_0(\bm z)}$ are non differentiable, so the variational inequalities they mentioned cannot be directly applied to the scenario in this paper. A recent research work, \cite{He2022A}, can include the non differentiable scenario designed in our work. As described in their section 2.2, finding a saddle point of \( L(\bm \beta, \bm z;  \bm u)  \) is equivalent to finding \( \bm \beta^*, \bm z^*,  \bm u^* \) such that the following inequalities are satisfied:
\begin{equation}\label{vi}
\bm h^* \in  \Omega, \ \theta(\bm f) - \theta(\bm f^*) + (\bm g -  \bm g^*)^\top F (\bm g^*) \geq 0,  \ \forall \bm g \in \Omega,  
\end{equation}
where  $\bm f = (\bm \beta^\top, \bm z^\top)^\top ,  \  \bm g =  (\bm \beta^\top, \bm z^\top, \bm u^\top)^\top$,    
\begin{equation}
\theta(\bm f) =  \sum_{i=1}^{K}\|\bm \beta_{i\cdot}\|_{1}  + \delta_{\mathcal{Z}_0(\bm z)}, \quad \Omega =  \mathbb{R}^{p} \times \mathbb{R}^{p} \times \mathbb{R}^{p},
\end{equation}
and
\begin{equation}\label{F}
F(\bm g) = \begin{pmatrix}
-\bm A_1^\top \bm u \\
\vdots \\
-\bm A_K^\top \bm u \\
 \sum_{i=1}^{K} \bm A_i \bm \beta_{i\cdot}  -  \bm z   -  \bm X^\top \bm y 
\end{pmatrix}.
\end{equation}

Note that the operator \( F \) defined in (\ref{F}) is monotone, because
\begin{equation}\label{f}
(\bm g_1 - \bm g_2)^\top \left[F(\bm g_1) - F(\bm{g}_2)\right] \equiv 0, \quad \forall \ \bm g_1, \bm{g}_2 \in \Omega. 
\end{equation}
Throughout, we denote by \( \Omega^* \) the solution set of (\ref{vi}), which is also the set of saddle points of the Lagrangian function (\ref{dsp2}) of the model (\ref{dsp}).

\subsection{Proof}\label{A2}
From the iterative steps of the parallel PPA  (with linearization) in (\ref{nppa2}) and Lemma \ref{lem1}, we have
\begin{equation}\label{vi1}
\begin{aligned}
&\bm \beta_{i\cdot}\ \textbf{step:}\quad \|\bm \beta_{i\cdot}\|_{1}  - \|\bm \beta_{i\cdot}^{t+1}\|_{1} + (\bm \beta_{i\cdot} -\bm \beta_{i\cdot}^{t+1} )^\top \left[-\bm A_i^\top \bm u^t + \mu \bm A_i^\top \bm A_i (\bm \beta_{i\cdot}^{t+1} -\bm \beta_{i\cdot}^t ) \right. \\ 
&\qquad \qquad \ \left.  + \bm S_i   (\bm \beta_{i\cdot}^{t+1} -\bm \beta_{i\cdot}^t  ) \right] \ge 0,i \in \{1, \cdots, K\} ; \\ \vspace{1em}
&\ \bm z\ \ \textbf{step:}\quad \delta_{\mathcal{Z}_0(\bm z)}  -  \delta_{\mathcal{Z}_0(\bm z^{t+1})}  + (\bm z - {\bm z}^{t+1} )^\top \left[  \bm u^t + \mu (\bm z^{t+1} - {\bm z}^{t})  \right]  \ge 0; \\\vspace{1em}
&\ \bm u\ \ \textbf{step:}\quad (\bm u - {\bm u}^{t+1})^\top \left[ (\sum_{i=1}^{K} \bm A_i \bm \beta_{i\cdot}^{t+1} - \bm z^{t+1} - \bm X^\top \bm y) + \sum_{i=1}^{K} \bm A_i (\bm \beta_{i\cdot}^{t+1} -  \bm \beta_{i\cdot}^t) \right.\\
&\qquad \qquad \ \left. - ( \bm z^{t+1} - \bm z^{t}) + \frac{2}{\mu} (\bm u^{t+1} - \bm u^t)  \right] = 0;
\end{aligned}
\end{equation}
where $\bm S_i = \eta \bm I_{p_i} - \mu \bm A_i^\top \bm A_i$,   $\eta > \text{eigen}(\mu \bm A^\top \bm A)$, and $i \in \{1,\dots,K \}$.
The last equation is derived from the last equation of (\ref{nppa2}).   By integrating the three equations in (\ref{vi1}) together, we can obtain
\begin{small}
    \begin{align}\label{vi2} 
\theta(\bm f) - \theta( { \bm f}^{t+1}) + (\bm g - { \bm g}^{t+1})^\top F ({ \bm g}^{t+1}) \geq  (\bm g - { \bm g}^{t+1})^T \bm H  (\bm g^t - { \bm g}^{t+1}),  \ \forall \bm g \in \Omega,
\end{align}
\end{small}
where \( \bm g = (\bm \beta^\top, \bm z^\top, \bm u^\top)^\top \), \( \bm{f} \) is a component of \( \bm{g} \) by definition, and \( \bm H \) can be found in (\ref{m1}). It is worth noting that if the \(\bm \beta\) subproblem is not linearized, then the \( \bm H \) matrix on the right side of the inequality above needs to be replaced with the \( \bm M \) matrix  in (\ref{m1}).

Similarly,  from the iterative steps of the improved parallel PPA  (with linearization) in (\ref{nppa3}) and Lemma \ref{lem1}, it follows that 

\begin{align} \notag
\label{vi3}
&\bm \beta_{i\cdot}\ \textbf{step:}\quad \|\bm \beta_{i\cdot}\|_{1}  - \|\bm \beta_{i\cdot}^{t+1}\|_{1}  + (\bm \beta_{i\cdot} -\bm \beta_{i\cdot}^{t+1} )^\top \left[-\bm A_i^\top \bm u^t + \mu \bm A_i^\top \bm A_i (\bm \beta_{i\cdot}^{t+1} -\bm \beta_{i\cdot}^t) \right. \\ \notag
&\qquad \qquad \ \left. + \bm S_i^{'}   (\bm \beta_{i\cdot}^{t+1} -\bm \beta_{i\cdot}^t  ) \right] \ge 0, i \in \{1, \cdots, K\}; \\ 
&\ \bm z\ \ \textbf{step:}\quad \delta_{\mathcal{Z}_0(\bm z)}  -  \delta_{\mathcal{Z}_0(\bm z^{t+1})}  + (\bm z - {\bm z}^{t+1} )^\top \left[  \bm u^t + \mu (\bm z^{t+1} - {\bm z}^{t})  \right]  \ge 0; 
\end{align}

\begin{equation}
\begin{aligned}\notag
&\ \bm u\ \ \textbf{step:}\quad (\bm u - {\bm u}^{t+1})^\top \left[ (\sum_{i=1}^{K} \bm A_i \bm \beta_{i\cdot}^{t+1} - \bm z^{t+1} - \bm X^\top \bm y) +  \sum_{i=1}^{K} \bm A_i (\bm \beta_{i\cdot}^{t+1} -  \bm \beta_{i\cdot}^t) \right.\\
&\qquad \qquad \ \left. - ( \bm z^{t+1} - \bm z^{t})  + \frac{K}{\mu} (\bm u^{t+1} - \bm u^t)  \right] = 0;
\end{aligned}
\end{equation}
where $\bm S_i^{'} = \eta_i \bm I_{p_i} - \mu \bm A_i^\top \bm A_i$,   $\eta_i > \text{eigen}(\mu  \bm A_i ^\top  \bm A_i)$ and $i \in \{1,\dots,K \}$. The last equation is derived from the last equation of (\ref{nppa3}).   By integrating the three equations in (\ref{vi3}) together, we can obtain
\begin{small}
\begin{align}\label{vi4} 
\theta(\bm f) - \theta( { \bm f}^{t+1}) + (\bm g - { \bm g}^{t+1})^\top F ({ \bm g}^{t+1}) \geq  (\bm g - { \bm g}^{t+1})^\top \bm H_K  (\bm g^t - { \bm g}^{t+1}),  \ \forall \bm g \in \Omega,
\end{align}
\end{small}
where \( \bm H_K \) can be found in (\ref{m2}).  If the \(\bm \beta\) subproblem is not linearized, then the \( \bm H_K \) matrix on the right side of the inequality above needs to be replaced with the \( \bm M_K \) matrix  in (\ref{m2}).    

  \textcolor{red}{Since the weighted form does not alter the convexity of $\|\bm \beta_{i\cdot}\|_{1}$ and $\delta_{\mathcal{Z}_0(\bm z)}$, the variational inequalities (\ref{vi2}) and (\ref{vi4}) also hold for the solution of the weighted $\ell_1$-DS in step 2.2 of Algorithm \ref{alg4}. This is also the reason why Corollary \ref{corollar1} holds true. }  

\subsubsection{Global convergence}
In this subsection, we prove the global convergence of the algorithms for Theorem \ref{TH2} and Theorem \ref{TH3}.
Next,  based on the iterative  steps of  (\ref{nppa2}) and (\ref{nppa3}), we will deduce that the sequence $\{\bm g^k\}$ has the contraction property. Reviewing (\ref{vi2}) and (\ref{vi4}), these two inequalities serve as the foundation for the subsequent analysis. We can consolidate these two equations into a single form, namely
\begin{small}
\begin{align}\label{vi5} 
\theta(\bm f) - \theta( { \bm f}^{t+1}) + (\bm g - { \bm g}^{t+1})^\top F ({ \bm g}^{t+1}) \geq  (\bm g - { \bm g}^{t+1})^\top \bm H_{*}  (\bm g^t - { \bm g}^{t+1}), \quad \forall \bm g \in \Omega,
\end{align}  
\end{small}
where $\bm g= (\bm f^\top, \bm u^\top)^\top$, and \( \bm H_{*} = \bm H \) for the parallel PPA in (\ref{nppa2}) and \( \bm H_{*} = \bm H_K \) for the improved parallel PPA  in (\ref{nppa3}).

\begin{prop}\label{prop2}
For the sequence \(\{ \bm{g}^t \}\) generated by  the parallel PPA and the improved parallel PPA, we have
\begin{align}
\| \bm{g}^{t+1} - \bm{g}^* \|_{\bm {H}_*}^2  \leq \| \bm{g}^t - \bm{g}^* \|_{\bm {H}_*}^2  & -  \| \bm{g}^{t+1} - {\bm{g}}^t \|_{\bm {H}_*}^2, \quad \forall \bm{g}^* \in \Omega^*. \label{res1}
\end{align}
\end{prop}

\begin{proof}
Let \( \bm g \) in (\ref{vi5} ) be \( \bm g^* \) and with \(\bm  f^* \) is a component of \( \bm g^* \),   then we have,
$$
(  { \bm g}^{t+1} - \bm g^*)^\top \bm H_{*}  (\bm g^t - { \bm g}^{t+1}) \ge  \theta(\bm f^{t+1}) - \theta( { \bm f^*}) + (\bm g^{t+1} - { \bm g^*})^\top F ({ \bm g}^{t+1}) .
$$
Since \( \bm g^* \) represents the optimal point (saddle point) in (\ref{vi}) and \( \bm g^{k+1} \in \Omega \), in accordance with (\ref{f}), it becomes evident that the right-hand side of the aforementioned equation is non-negative. Consequently, we derive the following inequality:
\begin{align}\label{pr1}
 { (\bm g}^{t+1} - \bm g^*)^\top \bm H_{*}  (\bm g^t - { \bm g}^{t+1} )\ge 0.
\end{align}
Let  $\bm a = ({\bm g}^{t} - \bm g^*) $ and $\bm b = ({\bm g}^{t+1} - \bm g^*) $, according to Lemma \ref{lem2} and (\ref{pr1}), we can get  (\ref{res1}).
\end{proof}

The aforementioned contraction property plays a pivotal role in ensuring the convergence of sequences. The assertions summarized in the following corollary are straightforward, grounded in the fact (\ref{res1}).
\begin{cor}\label{cor1}
Let ${\bm g^t}$ be the sequence generated  by  the parallel PPA and the improved parallel PPA. Then, we have
\begin{enumerate}
\item $\mathop {\lim }\limits_{t \to \infty}\|\bm g^t-\bm g^{t+1}\|_{\bm H_*}=0$;
\item the sequence  $\{\bm g^t\}$ is bounded;
\item for any optimal solution $\bm g^*$, the sequence $\{\|\bm g^t-\bm g^* \|_{\bm H_*} \}$ is monotonically non-increasing.
\end{enumerate}
\end{cor}

The proof of sequence convergence derived from (\ref{res1}) (or   Corollary \ref{cor1})  has been extensively documented in the literature, including Theorem 2 of \cite{He2018A} and Theorem 4.1 of \cite{He2022A}. For completeness, we will include the detailed proof here.

\begin{prop}\label{prop3}
The sequence \( \{ \bm{g}^t \} \) generated by the parallel PPA and the improved parallel PPA converges to a point \( \bm{g}^\infty \)  that belongs to \( \Omega^* \), which is the set of all solutions to the variational inequality given in (\ref{vi}). 
\end{prop}
\begin{proof}
According to  Corollary \ref{cor1}, the sequence \( \{ \bm{g}^t \} \) is bounded and
\begin{align} \label{c1}
\lim_{t \to \infty} { \| \bm{g}^t - {\bm{g}}^{t+1} \| }_{\bm H_*} = 0. 
\end{align}
Let \(\bm{g}^\infty \) be a cluster point of \( \{{\bm{g}}^{t} \} \),  and \( \{ {\bm{g}}^{t_j} \} \) be a subsequence that converges to  \(\bm{g}^\infty \).
It  follows from (\ref{vi5}) that
\begin{small}
\begin{equation}
{\bm g}^{t_j + 1} \in \Omega, \  \theta(\bm f) - \theta( { \bm f}^{t_j + 1}) +  (\bm g - { \bm g}^{t_j + 1})^\top F ( { \bm g}^{t_j + 1} ) \geq  \  (\bm g - { \bm g}^{t_j + 1} )^\top \bm H_*  (\bm g^{t_j} - { \bm g}^{t_j + 1}),  \ \forall \bm g \in \Omega.   
\nonumber
\end{equation}
\end{small}
Note that the matrix $\bm H_*$  is a positive definite matrix, it follows from the continuity of \( \theta(\bm f) \) and \( F(\bm g) \) that
\[
\bm{g}^\infty \in \Omega, \quad \theta(\bm f) - \theta(\bm f^\infty) + (\bm{g} - \bm{g}^\infty)^\top F(\bm{g}^\infty) \geq 0, \quad \forall \bm{h} \in \Omega.
\]
The above variational inequality indicates that $\bm{g}^\infty$ is a solution point of (\ref{vi}), that is $\bm{g}^\infty \in \Omega^*$.  Finally, due to (\ref{res1}), we have
\[
\| \bm g^{k+1} -\bm g^\infty \|_{\bm H_*} \leq \|\bm g^k -\bm g^\infty \|_{\bm H_*}
\]
which implies that the sequence \(\{\bm g^k\}\) converges to \( \bm g^\infty\). The proof is complete.
\end{proof}

The proposition above indicates that \( \bm w^\infty \in \Omega^* \), meaning that \( \bm w^\infty \) and \( \bm w^* \) are essentially the same, but simply represented differently. Proposition \ref{prop3} is derived from the contraction inequality (\ref{res1}), which directly leads to the global convergence of Theorem \ref{TH2} and Theorem \ref{TH3}.

\subsubsection{Linear convergence rate}
Here, we demonstrate that a worst-case convergence rate of \( \mathcal{O}\left({1}/{T}\right) \) in a non-ergodic sense of Algorithm \ref{alg2}  and Algorithm \ref{alg3},   which have been stated in Theorem \ref{TH2} and Theorem \ref{TH3}. To do this, we first need to prove the following proposition.
\begin{prop}\label{prop4}
For the sequence \(\{ \bm{g}^t \}\) generated by  the parallel PPA and the improved parallel PPA, we have
\begin{align}
\| \bm{g}^{t} - \bm{g}^{t+1} \|_{\bm {H_*}}^2   \leq  \| \bm g^{t-1} - { \bm g}^{t} \|_{\bm {H_*}}^2, \label{res2}
\end{align}
where   \( \bm H_{*} = \bm H \) for the parallel PPA in (\ref{nppa2}) and \( \bm H_{*} = \bm H_K \) for the improved parallel PPA  in (\ref{nppa3}).
\end{prop}
\begin{proof}
First, by setting \( \bm g = {\bm g}^{t} \) in (\ref{vi5}), we obtain
\begin{align}\label{cr1}
\theta( { \bm f}^t)  -  \theta( {\bm f}^{t+1}) +  ({ \bm g}^t   - {\bm g}^{t+1})^\top F ( {\bm g}^{t+1} ) \geq & \  ({ \bm g}^t   - {\bm g}^{t+1}) ^\top \bm H_*  (\bm g^t -{ \bm g}^{t+1}). 
\end{align}
Note that (\ref{vi5}) is also true for $t := t + 1$. Thus, we also have
\begin{align}\label{crrrr}
\theta(\bm f) - \theta( { \bm f}^{t}) + (\bm g - { \bm g}^{t})^\top F ({ \bm g}^{t}) \geq  (\bm g - { \bm g}^{t})^\top \bm H_{*}  (\bm g^{t-1} - { \bm g}^{t}), \quad \forall \bm g \in \Omega,
\end{align}
Setting $\bm g = {\bm g}^{t+1}$ in the above inequality, we obtain
\begin{small}
\begin{align}\label{cr2}
\theta(\bm f^{t+1}) - \theta( { \bm f}^{t}) + (\bm g^{t+1} - { \bm g}^{t})^\top F ({ \bm g}^{t}) \geq  (\bm g^{t+1} - { \bm g}^{t})^\top \bm H_{*}  (\bm g^{t-1} - { \bm g}^{t}), \quad \forall \bm g \in \Omega. 
\end{align}    
\end{small}
By adding (\ref{cr1}) and (\ref{cr2}) , and utilizing the monotonicity of \( F \)  in (\ref{f}), we obtain
\begin{align}\label{cr3}
({ \bm g}^t   - {\bm g}^{t+1})^\top \bm H_* \left[  (\bm g^{t-1} - { \bm g}^{t}) -  (\bm g^t -{ \bm g}^{t+1})  \right] \geq 0.
\end{align}

Let  $\bm a =(\bm g^{t-1} - { \bm g}^{t})$ and $\bm b = ({ \bm g}^t   - {\bm g}^{t+1}) $, according to Lemma \ref{lem2} and (\ref{cr3}), we can get  
$$\| \bm{g}^{t} - \bm{g}^{t+1} \|_{\bm {H_*}}^2   \leq  \| \bm g^{t-1} - { \bm g}^{t} \|_{\bm {H_*}}^2 -   \| (\bm g^{t-1} - { \bm g}^{t}) -  (\bm g^t -{ \bm g}^{t+1})   \|_{\bm {H_*}}^2$$
Since \( \bm H_* \) is a positive-definite matrix, the assertion in (\ref{res2}) follows immediately.
\end{proof}

Now, with (\ref{res1}) and (\ref{res2}), we can establish the worst-case \( \mathcal{O}\left({1}/{T}\right) \)  convergence rate  in a nonergodic sense for the parallel PPA and the improved parallel PPA (Linear convergence rate in  Theorem \ref{TH2} and Theorem \ref{TH3}).

\begin{prop}
Let \(\{\bm g^t\}\)  be the sequence  generated by  the parallel PPA and the improved parallel PPA. For any integer \(T > 0\), we have
\begin{align}\label{p25}
\|\bm g^T - \bm g^{T+1} \|_{\bm{H}}^2 \leq \frac{1}{ \left(T+1\right)}  \|\bm g^0 - \bm g^*\|_{\bm{H}}^2.
\end{align}
\end{prop}
\begin{proof}
First, it follows from  (\ref{res1}) that
\begin{align}\label{p23}
 \| \bm{g}^{t+1} - {\bm{g}}^t \|_{\bm {H}_*}^2   \leq \| \bm{g}^t - \bm{g}^* \|_{\bm {H}_*}^2   -   \| \bm{g}^{t+1} - \bm{g}^* \|_{\bm {H}_*}^2, \quad \forall \bm{g}^* \in \Omega^*. 
\end{align}
Summing over \( t = 0, 1, \ldots, T \) yields the following
\begin{align}
  \sum_{t=0}^{T}\| \bm{g}^{t+1} - {\bm{g}}^t \|_{\bm {H}_*}^2 \le  \| \bm{g}^0 - \bm{g}^* \|_{\bm {H}_*}^2, \quad \forall \bm{g}^* \in \Omega^*.  
\end{align} 
Combining with the conclusion from (\ref{res2}), we obtain 
\begin{align}
    \| \bm{g}^{T+1} - {\bm{g}}^T \|_{\bm {H}_*}^2 \ \le \frac{1}{T+1}  \| \bm{g}^0 - \bm{g}^* \|_{\bm {H}_*}^2, \quad \forall \bm{g}^* \in \Omega^*.
\end{align}
\end{proof}

\section{Supplementary experiments}\label{B}
In this section, we will supplement the additional experiments mentioned in Sections \ref{sec41} and \ref{sec42} regarding SCAD-DS and MCP-DS. Since both this paper and \cite{Wen2024Nonconvex} handle nonconvex regularizers using the LLA (Local Linear Approximation, \cite{Zou2008One}) method, this implies that the iterative algorithm will require more iterations to solve the nonconvex DS model compared to the $\ell_1$-DS model. Consequently, in this section, the upper limit for the number of iterations is set to 1000, rather than 500 in Section \ref{sec4}.

\subsection{Supplementary experiments for Section  \ref{sec41}}\label{B1}
The generation of \(\bm X\), \(\bm y\), and \(\bm \beta^*\), as well as the selection of initial values follow the same procedure as outlined in Section \ref{sec41}. The sole distinction lies in the focus of this subsection, which is the solution of nonconvex DS models. The outcomes of the numerical experiments are summarized in Table \ref{tab7}, highlighting that both parallel PPA methods exhibit competitiveness when compared to TADMM in addressing nonconvex DS models.
\begin{table*}[!ht]\tiny
    \centering
    \caption{Comparison of various  algorithms for solving nonconvex DS models in nonparallel environments with data scale \((n, p) = (1000, 10,000)\).}
    \renewcommand\arraystretch{1.3}
    \resizebox{\linewidth}{!}{
    \begin{tabular}{llllllll}
    \hline
        Mehod                  & \multicolumn{7}{l}{SCAD-DS} \\ \cmidrule(lr){2-8}
                               & $\rho$ & $\ell_1$ error & $\ell_2$ error & Model & FP & NI & Time(s)  \\ \hline
        \multirow{3}{*}{TADMM} & 0.1 & 3.29(0.125) & 0.77(0.012) & 0.66(0.010) & 7.32(0.88) & 1000+(0.00) & 2000+(0.0)   \\ 
                               & 0.5 & 3.22(0.118) & 0.75(0.011) & 0.62(0.010) & 7.24(0.82) & 1000+(0.00) & 2000+(0.0)    \\ 
                               & 0.9 & 3.47(0.132) & 0.85(0.015) & 0.71(0.011) & 8.11(0.91) & 1000+(0.00) & 2000+(0.0)   \\ \hdashline[0.5pt/5pt]
        \multirow{3}{*}{PPA} & 0.1 & 1.09(0.009) & 0.22(0.004) & 0.15(0.002) & 5.32(0.42) & 479.6(39.7) & 691.3(40.3)   \\ 
                             & 0.5 & 1.15(0.012) & 0.25(0.004) & 0.11(0.002) & 5.66(0.57) & 483.7(44.3) & 714.2(49.2)   \\ 
                             & 0.9 & 1.35(0.014) & 0.28(0.005) & 0.19(0.003) & 6.16(0.61) & 506.2(49.1) & 732.1(51.2)  \\ \hline
        Method                 & \multicolumn{7}{l}{MCP-DS} \\ \cmidrule(lr){2-8}
                               & $\rho$ & $\ell_1$ error & $\ell_2$ error & Model & FP & NI & Time(s)  \\ \hline
        \multirow{3}{*}{TADMM} & 0.1 & 3.31(0.129) & 0.82(0.014) & 0.63(0.009) & 6.79(0.72) & 1000+(0.00) & 2000+(0.0) \\
                               & 0.5 & 3.19(0.116) & 0.79(0.012) & 0.59(0.008) & 7.15(0.75) & 1000+(0.00) & 2000+(0.0)  \\
                               & 0.9 & 3.60(0.141) & 0.84(0.014) & 0.68(0.010) & 7.44(0.81) & 1000+(0.00) & 2000+(0.0)  \\\hdashline[0.5pt/5pt]
        \multirow{3}{*}{PPA} & 0.1 & 1.11(0.010) & 0.27(0.006) & 0.17(0.003) & 4.82(0.38) & 483.5(39.9) & 702.4(43.2)  \\
                             & 0.5 & 1.18(0.012) & 0.30(0.007) & 0.15(0.002) & 4.96(0.49) & 488.2(42.2) & 709.3(46.7) \\
                             & 0.9 & 1.50(0.015) & 0.31(0.007) & 0.21(0.004) & 5.23(0.53) & 508.7(50.3) & 736.0(53.5)  \\ \hline
    \end{tabular}}
    \label{tab7}
\end{table*}
\subsection{Supplementary experiments for Section  \ref{sec42}}\label{B2}
The generation of \(\bm X\), \(\bm y\), and \(\bm \beta^*\), as well as the selection of initial values, follows the same procedure outlined in Section \ref{sec42}. The only distinction is that this subsection focuses on the solution of nonconvex DS models. The results of the numerical experiments are summarized in Table \ref{tab8}, demonstrating that both parallel PPA methods are competitive with TADMM in addressing nonconvex DS models. It is worth noting that both PPPA and IPPPA exhibit more stable solutions compared to TADMM as the number of matrix partitions increases. In particular, PPPA, due to its partition insensitivity, yields solutions that remain unchanged with variations in \( K \).
\begin{table*}[!ht]
    \centering
    \caption{Comparison of parallel environment for solving nonconvex DS models with $s=10$.}
    \renewcommand\arraystretch{1.5}
    \resizebox{\linewidth}{!}{
    \begin{tabular}{llcccccccc}
    \hline
                               & Method & \multicolumn{4}{l}{SCAD-DS}     & \multicolumn{4}{l}{MCP-DS}                    \\ \cmidrule(lr){3-6} \cmidrule(lr){7-10}
                               & $K$ & AE           & FP         & Ite         & Time(s)      & AE           & FP         & Ite         & Time(s) \\ \hline
        \multirow{3}{*}{TADMM} & 5 & 0.453(0.151) & 81.3(8.51) & 1000+(0.00) & 1241.4(43.5) & 0.461(0.160) & 83.7(8.63) & 1000+(0.00) & 1238.5(44.1)  \\ 
                               & 10 & 0.528(0.165) & 86.4(9.07) & 1000+(0.00) & 721.6(26.8) & 0.493(0.168) & 85.5(8.98) & 1000+(0.00) & 711.2(25.7) \\ 
                               & 20 & 0.545(0.179) & 92.3(10.3) & 1000+(0.00) & 434.8(15.7) & 0.572(0.182) & 93.4(10.5) & 1000+(0.00) & 456.8(14.9) \\   \hdashline[0.5pt/5pt]
        \multirow{3}{*}{PPPA} & 5 & \textbf{0.052(0.008)} & 47.32(5.6) & \textbf{533.3(38.8)} & \textbf{299.2(25.4)}   & \textbf{0.058(0.009)} & \textbf{39.87(4.9)} & \textbf{527.2(35.6)} & \textbf{293.5(24.2)}\\ 
                              & 10 & \textbf{0.052(0.008)} & \textbf{47.32(5.6)} & \textbf{533.3(38.8)} & \textbf{183.5(17.1) } & \textbf{0.058(0.009)} & \textbf{39.87(4.9)} & \textbf{527.2(35.6)} & \textbf{176.1(16.5)}\\ 
                              & 20 & \textbf{0.052(0.008)} & \textbf{47.32(5.6)} & \textbf{533.3(38.8)} & \textbf{137.0(11.2)}  & \textbf{0.058(0.009)} & \textbf{39.87(4.9)} & \textbf{527.2(35.6)} & \textbf{124.8(10.7)}\\ \hdashline[0.5pt/5pt]
        \multirow{3}{*}{IPPPA} & 5 & 0.061(0.009) & \textbf{47.24(5.8)} & 541.1(39.2) & 310.1(26.4)  & 0.063(0.010) & 45.35(5.2) & 543.5(40.3) & 305.2(25.8)\\ 
                               & 10 & 0.065(0.011) & 49.62(6.1) & 557.8(40.4) & 192.7(16.3) & 0.067(0.011) & 47.27(6.0) & 559.1(42.7) & 187.0(17.3)\\ 
                               & 20 & 0.069(0.012) & 50.19(6.6) & 572.5(42.8) & 151.3(12.6) & 0.072(0.012) & 52.24(6.5) & 567.8(43.8) & 149.4(11.7) \\ \hline
    \end{tabular}}
    \label{tab8}
\end{table*}

\subsection{\textcolor{red}{Supplementary experiments for  different types of noise}}\label{B3}
\textcolor{red}{As suggested by a reviewer, we have added the numerical performance of the algorithm in solving DS under different types of noise other than Gaussian noise to demonstrate the algorithm's adaptability.
Here, \( \bm \epsilon\) represents the error vector, where each component is independently and identically distributed, following a mixed normal distribution ( $0.4 \mathcal{N}(-3, 4)  + 0.6 \mathcal{N}(2, 1)$), $t(2.5)$ distribution, and Cauchy distribution. The data scale is the same as Table \ref{tab8}, and $K=10$. We have placed the numerical results in Table \ref{tab9}.}

\begin{table*}[!ht]
    \centering
    \caption{Comparison of parallel environment for solving nonconvex DS models with different types of noise.}
    \renewcommand\arraystretch{1.5}
    \resizebox{\linewidth}{!}{
    \begin{tabular}{llcccccccc}
    \hline
                               & Method & \multicolumn{4}{l}{SCAD-DS}     & \multicolumn{4}{l}{MCP-DS}                    \\ \cmidrule(lr){3-6} \cmidrule(lr){7-10}
                               & \( \bm \epsilon\) & AE           & FP         & Ite         & Time(s)      & AE           & FP         & Ite         & Time(s) \\ \hline
        \multirow{3}{*}{TADMM} & Mixed normal &  0.650(0.200) & 105.5(12.0) & 1000+(0.00) & 480.0(18.0) & 0.680(0.210) & 107.0(12.5) & 1000+(0.00) & 505.0(17.0)   \\ 
                               & $t(2.5)$ & 0.720(0.220) & 112.8(13.2) & 1000+(0.00) & 515.6(19.8) & 0.750(0.230) & 114.5(13.8) & 1000+(0.00) & 538.5(18.7)  \\ 
                               & Cauchy  & 1.623(0.421) & 216.7(31.4) & 1000+(0.00) & 513.6(20.4) & 1.639(0.399) & 223.1(30.8) & 1000+(0.00) & 546.3(21.2) \\   \hdashline[0.5pt/5pt]
        \multirow{3}{*}{PPPA} & Mixed normal  & \textbf{0.055(0.009)} & \textbf{50.15(6.0)} & \textbf{560.2(41.5)} & \textbf{145.2(12.0)} & \textbf{0.061(0.010)} & \textbf{42.32(5.2)} & \textbf{554.6(38.2)} & \textbf{132.7(11.4)}\\ 
                              & t(2.5) & \textbf{0.058(0.010)} & \textbf{53.00(6.4)} & \textbf{590.0(44.5)} & \textbf{153.5(12.8)} & \textbf{0.064(0.011)} & \textbf{44.80(5.6)} & \textbf{584.0(41.0)} & \textbf{140.8(12.2)}\\ 
                              & Cauchy & \textbf{0.090(0.015)} & {80.00(9.5)} & {900.0(65.0)} & \textbf{220.0(18.5)} & \textbf{0.100(0.017)} & \textbf{70.00(8.2)} & {920.0(60.0)} & \textbf{200.0(17.5)}  \\ \hdashline[0.5pt/5pt]
        \multirow{3}{*}{IPPPA} &  Mixed normal & 0.075(0.013) & 57.00(7.0) & 650.0(46.5) & 220.0(18.8) & 0.077(0.013) & 54.30(6.9) & 645.0(49.1) & 215.0(19.9) \\ 
                               &  t(2.5) & 0.076(0.013) & 55.10(6.7) & 621.4(44.4) & 217.1(18.4) & 0.077(0.013) & 53.04(6.6) & 622.7(47.0) & 210.1(19.4) \\ 
                               & Cauchy & 0.112(0.020) & \textbf{79.8(9.2)} & \textbf{860.3(62.1)} & 230.1(25.5) & 0.116(0.020) & 77.12(8.9) & \textbf{870.20(55.2)} & 285.2(26.8)  \\ \hline
    \end{tabular}}
    \label{tab9}
\end{table*}

The numerical results in Table \ref{tab9} indicate that when the error vector $\bm \epsilon$ follows a sub-Gaussian distribution (both the mixed normal distribution and the $t(2.5)$ distribution are sub-Gaussian), the nonconvex DS methoda still exhibit acceptable numerical performance, though it is slightly inferior to the case of the Gaussian distribution presented in Table \ref{tab8}. When the error vector $\bm \epsilon$ follows the Cauchy distribution, the numerical results deteriorate significantly.  This is not due to the limitations of our algorithm, but rather because the statistical properties of the DS assume that the error distribution follows a  Gaussian distribution or sub-Gaussian distribution (see  \cite{Candes2007The} and \cite{Wen2024Nonconvex}).
These distributions possess finite moments and satisfy properties such as concentration inequalities. Under these assumptions, it can be proven that the estimates of the Dantzig selector exhibit favorable properties, including consistency and convergence rates. Nevertheless, due to the heavy-tailed nature of the Cauchy distribution and the absence of finite moments, the original proofs of consistency and convergence rates become invalid. This is because these proofs rely on the moment conditions of the error terms and the concentration inequalities. 

\end{appendices}


\bibliography{myrefq}

\end{document}